\documentclass{article}

\usepackage{microtype}
\usepackage{graphicx}
\usepackage{subfigure}
\usepackage{booktabs} %

\usepackage[colorlinks=true, allcolors=blue]{hyperref}

\usepackage[accepted]{icml2024}

\usepackage{amsmath}
\usepackage{amssymb}
\usepackage{mathtools}
\usepackage{amsthm}

\usepackage[capitalize,noabbrev]{cleveref}

\theoremstyle{plain}
\newtheorem{theorem}{Theorem}[section]

\newtheorem{lemma}[theorem]{Lemma}

\theoremstyle{definition}
\newtheorem{definition}[theorem]{Definition}

\theoremstyle{remark}

\usepackage[colorinlistoftodos,prependcaption,textsize=tiny]{todonotes}

\usepackage{bbm}
\usepackage{dsfont}
\usepackage{xspace}
\usepackage{accents}
\usepackage{xparse}
\usepackage{enumitem}
\newcommand{\alglinelabel}{%
  \addtocounter{ALC@line}{-1}%
  \refstepcounter{ALC@line}%
  \label%
}

\newtheorem{example}[theorem]{Example}

\newtheorem{fact}[theorem]{Fact}

\renewcommand{\algorithmiccomment}[1]{$\triangleright $ #1}

\newcommand{\cA}{\mathcal{A}}

\newcommand{\cD}{\mathcal{D}}
\newcommand{\cE}{\mathcal{E}}

\newcommand{\cR}{\mathcal{R}}

\newcommand{\cU}{\mathcal{U}}

\newcommand{\cY}{\mathcal{Y}}
\newcommand{\cZ}{\mathcal{Z}}

\newcommand{\C}{\mathbb{C}}
\newcommand{\N}{\mathbb{N}}
\newcommand{\R}{\mathbb{R}}

\newcommand{\Z}{\mathbb{Z}}

\newcommand{\norm}[1]{\left\| {#1} \right\|}
\newcommand{\PAREN}[1]{{\left( #1 \right)}}
\newcommand{\paren}[1]{{( {#1} )}}
\newcommand{\bigparen}[1]{{\big( {#1} \big)}}

\newcommand{\bracket}[1]{{\left[ {#1} \right]}}
\newcommand{\card}[1]{\left| {#1} \right|}
\newcommand{\set}[1]{\left\{ {#1} \right\}}
\newcommand{\eps}{\varepsilon}

\NewDocumentCommand\p{ m g }{
  \ensuremath{
    \IfNoValueTF{#2}
    {\Pr [ #1 ]}
    {\Pr_{#1}[#2]}
  }
}
\RenewDocumentCommand\P{ m g }{
  \ensuremath{
    \IfNoValueTF{#2}
    {\Pr \left[#1\right]}
    {\Pr_{#1}\left[#2\right]}
  }
}
\DeclareMathOperator*{\Expectation}{\mathbb{E}}
\NewDocumentCommand\E{ m g }{
  \ensuremath{
    \IfNoValueTF{#2}
    {\Expectation \left[#1\right]}
    {\Expectation_{#1}\left[#2\right]}
  }
}
\DeclareMathOperator*{\Variance}{\mathbb{V}\mathrm{ar}}
\NewDocumentCommand\Var{ m g }{
  \ensuremath{
    \IfNoValueTF{#2}
    {\Variance \left[#1\right]}
    {\Variance_{#1}\left[#2\right]}
  }
}

\NewDocumentCommand\GammaDist{ m g }{
  \ensuremath{
    \IfNoValueTF{#2}
    {\operatorname{\mathcal{G}amma}\left( {#1} \right)}
    {\operatorname{\mathcal{G}amma}\left( {#1}, {#2} \right)}
  }
}

\NewDocumentCommand\LapNoise{ m g }{
  \ensuremath{
    \IfNoValueTF{#2}
    {\operatorname{\mathbb{L}ap}\left( {#1} \right)}
    {\operatorname{\mathbb{L}ap}\left( {#1}, {#2} \right)}
  }
}

\NewDocumentCommand\GumbelNoise{ m g }{
  \ensuremath{
    \IfNoValueTF{#2}
    {\operatorname{\mathbb{G}umbel}\left( {#1} \right)}
    {\operatorname{\mathbb{G}umbel}\left( {#1}, {#2} \right)}
  }
}

\newcommand{\algoFstInv}{\mathcal{A}_{\textit{fst-inv}}}
\newcommand{\algoRounding}{\mathcal{A}_{\textit{rnd}}}

\newcommand{\hist}{\vec{h}}
\newcommand{\err}[1]{\cE\textit{rr}_{#1}}

\newcommand{\privHist}{\tilde{h}}
\newcommand{\diag}[1]{\operatorname{\bf diag} \PAREN{#1}}
\newcommand{\Lb}{{\footnotesize \text{LB}}}
\newcommand{\Ub}{{\footnotesize \text{UB}}}
\newcommand{\DataDomain}{\mathcal{D}}
\newcommand{\viewA}{\Pi_{\vec{x}}}
\newcommand{\viewB}{\Pi_{\vec{y}}}
\newcommand{\ProbNorm}{P_{norm}}
\newcommand{\TransMatrix}{\pmb{A}}
\newcommand{\CirculantMatrix}[1]{\pmb{\text{circ}}\PAREN{#1}}
\newcommand{\profile}{\vec{f}}

\newcommand{\boundedProfileOfNoisyHist}{\tilde{f}}
\newcommand{\largeSet}{T_{> 1}}
\newcommand{\surplus}{s_{> 1}}
\newcommand{\smallSet}{T_{< 0}}
\newcommand{\deficit}{s_{< 0}}
\newcommand{\estProfile}{\vec{r}}
\newcommand{\modularIndex}[1]{\mathsf{mi}\PAREN{{#1}}}
\newcommand{\DFTMatrix}{\pmb{V}}
\newcommand{\IntSet}[2]{[{#1}\,.\,.\,{#2}]}
\newcommand{\indicator}[1]{\mathds{1}_{\left[#1\right]}}
\newcommand{\vecAllOne}{\vec{\mathbbm{1}}}
\newcommand{\eigenvalue}{\varphi}
\newcommand{\vecTargetSupport}{\vec{\mathbbm{1}}_{{0:n}}}
\newcommand\inner[2]{\left\langle #1, #2 \right\rangle}
\newcommand\tinner[2]{\langle #1, #2 \rangle}
\newcommand\clip[3]{\text{clip} \bracket{ {#1}, {#2}, {#3} } }

\NewDocumentCommand\DiscreteLapNoise{ m g }{
  \ensuremath{
    \IfNoValueTF{#2}
    {\operatorname{\mathbb{DL}ap}\left( {#1} \right)}
    {\operatorname{\mathbb{DL}ap}\left( {#1}, {#2} \right)}
  }
}
\NewDocumentCommand\TDiscreteLapNoise{ m g }{
  \ensuremath{
    \IfNoValueTF{#2}
    {\operatorname{\mathbb{TDL}ap}\left( {#1} \right)}
    {\operatorname{\mathbb{TDL}ap}\left( {#1}, {#2} \right)}
  }
}
\NewDocumentCommand\ExpNoise{ m g }{
  \ensuremath{
    \IfNoValueTF{#2}
    {\operatorname{\mathbb{DE}xp}\left( {#1} \right)}
    {\operatorname{\mathbb{DE}xp}\left( {#1}, {#2} \right)}
  }
}

\newif\ifcomment
\commentfalse

\definecolor{DarkGreen}{rgb}{0.1,0.5,0.1}
\ifcomment
\newcommand{\hao}[1]{\textcolor{blue}{[HAO: #1]}}

\else

\makeatletter
\newcommand{\hao}[1]{%
  \@bsphack
  \@esphack
}
\makeatother

\fi

\begin{document}

\twocolumn[
\icmltitle{Profile Reconstruction from Private Sketches}

\icmlsetsymbol{equal}{*}

\begin{icmlauthorlist}
\icmlauthor{Hao WU}{yyy}
\icmlauthor{Rasmus Pagh}{yyy}
\end{icmlauthorlist}

\icmlaffiliation{yyy}{Department of Computer Science, University of Copenhagen, Denmark}

\icmlcorrespondingauthor{Hao WU}{hawu@di.ku.dk}
\icmlcorrespondingauthor{Rasmus Pagh}{pagh@di.ku.dk}

\icmlkeywords{Machine Learning, ICML}

\vskip 0.3in
]

\printAffiliationsAndNotice{}  %

\begin{abstract}
Given a multiset of $n$ items from $\DataDomain$, the \emph{profile reconstruction} problem is to estimate, for $t = 0, 1, \dots, n$, the fraction $\profile[t]$ of items in $\DataDomain$ that appear exactly $t$ times.
We consider differentially private profile estimation in a distributed, space-constrained setting where we wish to maintain an updatable, private sketch of the multiset that allows us to compute an approximation of $\profile = (\profile[0], \dots, \profile[n])$.
Using a histogram privatized using discrete Laplace noise, we show how to ``reverse'' the noise, using an approach of Dwork et al.~(ITCS '10).
We show how to speed up their LP-based technique from polynomial time to $O(d + n \log n)$, where $d = |\DataDomain|$, and analyze the achievable error in the $\ell_1$, $\ell_2$ and $\ell_\infty$ norms.
In all cases the dependency of the error on $d$ is $O( 1 / \sqrt{d})$ --- we give an information-theoretic lower bound showing that this dependence on $d$ is asymptotically optimal among all private, updatable sketches for the profile reconstruction problem with a high-probability error guarantee.
\end{abstract}

\section{Introduction}\label{sec:introduction}

The \emph{profile} is a fundamental statistic of a data set.
Given a multiset of $n$ items from a finite domain $\DataDomain$, the profile is a vector $\profile$ whose $t^{(th)}$ entry equals the fraction $\profile[t]$ of items in $\DataDomain$ that appear exactly $t$ times.
The concept has arisen in many contexts under various names, including \emph{rarity}, \emph{incidence}, \emph{fingerprint}, \emph{prevalence}, \emph{collision statistics}, \emph{anonymous histogram}, \emph{unattributed histogram}, and \emph{histogram of histograms}.
Besides succinctly summarizing frequency information in a dataset, for example the degree distribution in a graph, a profile is sufficient to compute any \emph{symmetric} function of the dataset histogram~\citep{chen_et_al:LIPIcs.ITCS.2024.32}.

In settings where space, communication or privacy constraints prevents us from computing the exact profile it is of interest to \emph{approximate} the profile $\profile$ by a vector $\estProfile$ such that $\estProfile - \profile$ has small $\ell_p$ norm for some parameter $p$.

{\bf Streaming algorithms.} \citet{datar2002estimating} considered estimating the profile in a streaming setting, showing that a small sample of the domain $\DataDomain$ suffices to ensure small error (they focus on $\ell_1$ error, but the technique applies to general $\ell_p$ norms).
The sample can be described by a compact hash function, making it suitable for use in distributed settings.
More recently,~\citet*{chen_et_al:LIPIcs.ITCS.2024.32} improved the error bound, and showed matching space lower bounds for streaming algorithms.
General estimators for symmetric functions based on samples were presented and analyzed by \citet*{icmlAcharyaDOS17}. %

{\bf Differentially private algorithms.}
Differential privacy~\cite{DworkMNS06} is a leading approach to ensure that the output of an algorithm, from simple statistics to machine learning models, does not reveal too much about the inputs.
It is one of the central techniques for enabling collection of statistics, training of machine learning models, and other computations on data that may hold sensitive information.
We will focus on the privacy of individual elements in the multiset, that is, our differential privacy guarantee is defined in terms of pairs of neighboring datasets where the number of occurrences of a particular element differs by 1, and that are otherwise identical.

In the \emph{central model} of differential privacy (where a trusted curator holds the input data), privately releasing a \emph{profile} with small $\ell_2$ error was considered by \citet*{HayLMJ09, HayRMS10}, the latter focusing on the application of graph degree distributions. %
Private synthetic graphs with specific degree distribution were also studied by \citet*{KarwaS12}. %

Motivated by estimating repetitions in password datasets, \citet*{BlockiDB16} studied differentially private protocols for profile estimation, focusing on $\ell_2$ error.

\citet*{Suresh19} considered an alternative (incomparable) error metric, where the profile is represented in the ``verbatim'' form of a vector encoding all frequencies in sorted order. They presented an algorithm to release verbatim profiles with a small $\ell_1$ error. Recently \citet*{Manurangsi22} showed improved tight upper and lower bounds on the privacy/$\ell_1$ error trade-off for verbatim profiles.

{\bf Differentially private distributed/streaming algorithms}
\citet*{DworkNPRY10} %
studied profile estimation (referred to as \emph{$t$-incidence}) in the pan-private streaming setting, where the internal state of the streaming algorithm must satisfy a differential privacy guarantee.
Pan-private algorithms protect against \emph{intrusions}, i.e., unintended release of the internal state of the algorithm.
Unlike the other works on private profile estimation mentioned above, the neighboring relation considered by~\citet{DworkNPRY10} is defined on the profile itself, i.e., datasets are neighboring if they differ \emph{by any amount} on the count of one item (so the profiles differ in two entries).
In a nutshell%
\footnote{Since there is no good bound on the difference between histograms with their notion of neighboring datasets, ~\citet{DworkNPRY10} actually work with histograms of counts modulo several small random primes $p_i$, and use these to estimate profiles on the counts modulo $p_i$, which are in turn used to estimate the original profile. In our context this complication goes away.},
their algorithm works by maintaining a \emph{histogram} that stores the number of occurrences of each item perturbed by additive noise.
Since the noisy histogram is a linear function of the histogram it is \emph{updatable} when the count of an item changes.
\citet{DworkNPRY10} finally show how to derive an estimate of the profile from the noisy histogram by solving a linear programming problem.
The size of the sketch required for achieving $\ell_\infty$ error $\alpha$ on the profile estimate with high probability is stated as $\text{poly}(\log n, 1/\eps, 1/\alpha )$ with unspecified exponents.
Unfortunately, the paper does not analyze for what values of $\alpha$ it is possible to solve the linear program (with high probability), meaning that the privacy-utility trade-off obtained remains unclear.

\subsection{Our contributions}

In this paper we study how to approximate the profile from a noisy histogram of the data.
Specifically, we consider the setting where \emph{discrete Laplace} noise has been added to the counts in the histogram.
(This is different from the noise distribution used by~\citet{DworkNPRY10}, which is nearly uniform modulo a small prime~$p_i$.) 
Computing the empirical profile directly from the noisy histogram can lead to large errors.
For example, consider the situation where $\DataDomain$ has $n$ items that each appear once in the dataset and let $\eps > 0$ be a small constant.
After adding noise to the histogram to make it $\eps$-differentially private (and possibly rounding up negative values to 0), the fraction of items with a count of~1 will be $\approx \eps$, far from the true fraction of $1$.

Instead we use an approach similar to that of~\citet{DworkNPRY10} to approximately invert the effect of the noise.
Since we use a known noise distribution, the \emph{expected} value of the profile computed from the noisy histogram is a linear function of the true profile.
Due to concentration we can argue that the observed values in the profile based on the noisy histogram are in fact close to the expected ones, and thus it makes sense to find a histogram whose expected profile after noise would be close to the observed noisy profile.
Unlike~\citet{DworkNPRY10} who formulate this as a feasibility question of a linear program, we show that it can be reduced to solving a linear $\ell_p$-regression problem with a circulant matrix, paving the way for a near-linear time algorithm.
We also give a careful analysis of the error $\err{p}$ in $\ell_p$ norm, for $p\in\{1,2,\infty\}$, showing the following upper bound:

\begin{theorem} \label{theorem: property of profile approximation algorithm}
    Let $\eta \in (0, 1)$, $\eps > 0$. 
    Denote $B \doteq \frac{1}{\eps} \ln \big( \max \set{ \frac{2d}{\eta(e^\eps + 1)}, \frac{8 \, e^\eps}{e^{2\eps} - 1} } \big)$, and assume that~$n \ge B$.
    There is an algorithm that, given a private version~$\privHist$ of a $d$-dimensional histogram~$\hist$ published by the~$\eps$-DP discrete Laplace mechanism and a parameter~$p = 1, 2$ or~$\infty$, returns an approximate~$\estProfile$ for the profile~$\profile$ of~$\hist$ in~$O(d + n \log n)$ time. 
    Moreover, with probability at least $1 - \eta$, the estimation error~$\err{p} \doteq || \estProfile - \profile ||_p$ satisfies
    
    \vspace{-4mm}
    \resizebox{\linewidth}{!}{
      \begin{minipage}{1.2\linewidth}
        \begin{align}
            \err{1} 
            &\in O \PAREN{ 
                    \PAREN{
                        \frac{ 
                            e^{ \eps } + 1
                        }{ 
                            e^{ \eps } - 1
                        }
                    }^2 \PAREN{  
                        \frac{
                            \sum_{t = -B}^{n + B} \sqrt{ \E{ \tilde{f} [t] } }
                        }{ \sqrt{d} } 
                        + \sqrt{
                            \frac{2\ln \frac{1}{\eta}}{d}
                        }
                    }
                }
            \label{ineq: error 1}
            \\
            \err{2}
            &\in O \PAREN{  
                    \PAREN{
                        \frac{ 
                            e^{ \eps } + 1
                        }{ 
                            e^{ \eps } - 1
                        }
                    }^2 
                    \sqrt{
                        \frac{\ln \frac{1}{\eta}}{d}
                    } 
                }
            \label{ineq: error 2}
            \\
            \err{\infty}
                &\in \PAREN{ 
                    \PAREN{
                       \frac{ 
                            e^{ \eps } + 1
                        }{ 
                            e^{ \eps } - 1
                        }
                    }^2  
                    \PAREN{
                        \sqrt{ 
                            \frac{ 
                                2 \, \PAREN{e^{ \eps } - 1}
                            }{  
                                d \,\PAREN{e^{ \eps } + 1}
                            } \ln \frac{n}{\eta} 
                        } 
                        + 
                        \frac{\ln \frac{n}{\eta}}{ 3 \, d }  
                    }
                }  
            \label{ineq: error 3}
        \end{align}
      \end{minipage}
    }
    
    \vspace{-1mm}
    where 
    $
        \tilde{f} [t] = %
            \frac{1}{d} | \{\ell \in \DataDomain : \privHist[\ell] = t \} |,
    $
    ~$t = \text{-}B, \ldots, n + B.$
\end{theorem}

\noindent
Notice that the error bounds in all three norms are proportional to $O(1/\sqrt{d})$, where $d$ is the domain size.
As we will discuss below, this dependence is optimal under $\varepsilon$-differential privacy for updatable sketches.
On the other hand, 
when~$\eps \in O(1)$,
the $O(1 / \eps^{1.5} )$ dependence for~$\err{\infty}$
on the privacy parameter does not match the $\Omega(1/\varepsilon)$ lower bound on the error of releasing a single counter~\cite{HardtT10}.
We note that the bound on $\err{1}$ is a refinement of the bound implied by the bound on $\err{2}$ and the fact that $\ell_1$ distances in $n$ dimensions are at most $\sqrt{n}$ times larger than $\ell_2$ distances.
Since errors are normalized by a $1/d$ factor one might hope to achieve errors that decrease more rapidly with $d$ than those of Theorem~\ref{theorem: property of profile approximation algorithm}.
However, we show the following lower bound:  

\begin{theorem}
    \label{theorem: accuracy lower bound of profile estimation}
    For every $\eps$-DP and updatable algorithm~$\cA : \IntSet{0}{n}^d \rightarrow \cY$, and every profile estimation algorithm~$\cR : \cY \rightarrow [0, 1]^{n + 1}$, 
    if the sensitivity~$\Delta_{\cR \circ \cA} \in o \big( \frac{1}{\sqrt{d} } \cdot \frac{\eps}{e^{c \cdot \eps}} \big)$ (for some universal constant $c > 0$), then
    there exists an input~$\hist \in \IntSet{0}{n}^d$ such that
    \begin{equation} 
        \label{eq: accuray lower bound}
        \begin{array}{c}
            \E{ \norm{ \profile - \cR \PAREN{ \cA(\hist) } }_\infty }
                \in \Omega \Big( \frac{1}{\sqrt{d} \cdot e^{c \cdot \eps}} \Big).
        \end{array}
    \end{equation}
\end{theorem}

This lower bound is not limited to algorithms reconstructing the profile from a perturbed histogram published by the discrete Laplace mechanism. 
Instead, it applies to a broader class of algorithms that attempt to reconstruct the profile from outputs of \emph{updatable} $\eps$-differentially private algorithms.
The formal definition of an ``updatable'' algorithm and the sensitivity~$\Delta_{\cR \circ \cA}$ will be given in Section~\ref{sec: lower bound}.
Informally, the former denotes the property that when the input histogram $\hist$ is updated, we can directly modify $\cA(\hist)$ to reflect the change, and the latter represents the maximum distance between the profiles reconstructed from the outputs of $\cA$ for neighboring histograms.
The lower bound shows that an error of~$\Omega (1 / \sqrt{d})$ is unavoidable for such sketches.

\vspace{-2mm}
\subsection{Technical overview}
As noted by~\citet{datar2002estimating} and~\citet{DworkNPRY10} it suffices to estimate the profile of a sample of elements from the domain, since the profile of the sample will closely resemble the true profile.
Thus we will assume that this sampling step has already been done, and consider the full histogram of (sampled) items.
The final error bounds add the sampling error and the error from privacy using the triangle inequality, but here we focus on the latter.

Given a noisy histogram~$\privHist$, we can compute an empirical profile~$\boundedProfileOfNoisyHist$.
As the noisy histogram is obtained by adding independent noise to each coordinate of the full histogram~$\hist$, the empirical profile can be viewed as a function of a set of independent random variables and should concentrate around its expectation~$\mathbb{E}[\boundedProfileOfNoisyHist]$.
Hence, we regard the empirical profile as a reliable representative of its expectation.

Although the expectation~$\mathbb{E}[\boundedProfileOfNoisyHist]$ itself may not be a good approximation of the original profile~$\profile$ of~$\hist$, there is a subtle connection between them:~$\mathbb{E}[\boundedProfileOfNoisyHist]$ can be computed via a linear transform of the original profile~$\profile$.
Due the symmetry of discrete Laplace noise, we can write the linear transform as~$\TransMatrix \profile$ for a circulant matrix~$\TransMatrix$ with non-zero eigenvalues.

Since the inversion for such linear transforms can be computed in~$O(n \log n)$ time via Fast Fourier transform (FFT), it motivates us to compute~$\TransMatrix^{-1} \boundedProfileOfNoisyHist$ as an approximation of the true profile.
However, such an approximation may not be a valid profile: it may contain coordinates outside~$[0, 1]$, and its coordinates do not sum up to~$1$.
We will design a systematic procedure for fixing these violations, which also relies on FFT and runs in~$O(n \log n)$ time.

The approximation error comes from three sources: the deviation of~$\boundedProfileOfNoisyHist$ from its expectation, the inverse linear transform, and the violation fixing procedure. 
Each part will be analyzed separately. 
The second part, in particular, might be of independent interest as it involves examining various matrix norms of~$\TransMatrix^{-1}$, which naturally arises from the discrete Laplace mechanism.

To establish a lower bound, we employ a reduction from the $d$-dimensional inner product estimation problem in the framework of two-party differential privacy where an error of~$\Omega(\sqrt{d})$ is inevitable~\citep{McGregorMPRTV10, HaitnerMST22}.

\vspace{-3mm}
\section{Problem Description}

Let~$\DataDomain \doteq \set{1, \ldots, d}$ be a set of~$d$ items.
Given a multiset of~$n$ items from $\DataDomain$,
for each item~$\ell \in \DataDomain$, let~$\hist[\ell]$ be the number of times it appears in the multiset.
The \emph{histogram} is a vector~$\hist \doteq \paren{ \hist[1], \ldots, \hist[d] } \in \IntSet{0}{n}^d$. 
The \emph{profile} of the histogram~$\hist$ is a frequency vector~$\profile \doteq (\profile[0], \ldots, \profile[n]) \in [0,1]^{n + 1},$ 
where 
$
    \profile[t] = 
        \PAREN{1 / d} \cdot | \{ \ell \in \cD : \hist[\ell] = t \} |
$
is the fraction of elements in $\cD$ which appears~$t$ times in the multiset, for each $t \in \IntSet{0}{n}$.

Our goal is to design an efficient and accurate algorithm for estimating the profile~$\profile$ from a private histogram~$\privHist$ output by the \emph{discrete Laplace mechanism} (or geometric mechanism) due to~\citet*{GhoshRS09}.

\vspace{-3mm}
\paragraph{Discrete Laplace Mechanism.} 
    Given an input histogram~$\hist$, the mechanism aims at outputting a perturbed version~$\privHist$ of~$\hist$ that is insensitive to 
    the addition/removal of an arbitrary item in the multiset.
    Specifically, the mechanism creates~$\privHist$ by adding independent noises to each coordinate of~$\hist$, such that
    \vspace{-1mm}
    \begin{equation*}
        \privHist[\ell] = \hist[\ell] + \DiscreteLapNoise{e^{-\eps}},\, 
        \forall \ell \in \DataDomain,
    \end{equation*}
    where $\DiscreteLapNoise{e^{-\eps}}$ denotes a random variable following discrete Laplace distribution: 
    $
        \P{ \DiscreteLapNoise{e^{-\eps}} = t } = \frac{1 - e^{-\eps}}{1 + e^{-\eps}} \cdot e^{-\eps \cdot \card{t} }, \, \forall t \in \Z.
    $
    In practice, a clipping step can be applied after coordinate perturbation:
    \vspace{-1mm}
    \begin{equation*}
        \privHist[\ell] = \clip{\hist[\ell] + \DiscreteLapNoise{e^{-\eps}}}{0}{n},\, 
        \forall \ell \in \DataDomain.
    \end{equation*}
    where~$\text{clip} [\hist[\ell] + \DiscreteLapNoise{e^{-\eps}}, 0, n] \doteq \max \{ 0, \min \{ \hist[\ell]$ $+ \DiscreteLapNoise{e^{-\eps}}, n \} \}$ ensures the returned value is in~$[0 \,.\,.\, n]$. 

    {\it Privacy Guarantee.}
    The privacy guarantee provided by the mechanism is measured by the framework of \emph{differential privacy}.
    Formally, we call two histograms~$\hist$ and~$\hist'$ \emph{neighboring}, written as~$\hist \sim \hist'$, 
    if 
    $\exists \ell \in \DataDomain$, s.t.,~$| \hist[\ell] - \hist'[\ell] | \le 1$, and $\forall \ell' \in \cD \setminus \set{\ell}$, it holds that $\hist[\ell'] = \hist'[\ell']$. 
    The discrete Laplace mechanism ensures that the output histograms have similar distributions for neighboring inputs.

\begin{definition}[${\eps}$-Differentially Privacy~\citep{DR14}] \label{def: Differential Privacy}
    Given~$\eps > 0$, a randomized algorithm 
    $\cA: \N^d \rightarrow \cZ$ 
    is called~${\eps}$-differentially private (DP),
    if for every~$\hist, \hist' \in \N^d$ such that~$\hist \sim \hist'$, 
    and 
    all (measurable) $Z \subseteq \cZ$,
    \begin{equation} \label{ineq: def private algo}
        \begin{array}{c}
            \P{ \cA (\hist) \in Z } \le e^\eps \cdot \Pr [ \cA (\hist') \in Z ]\,.
        \end{array}
    \end{equation}
    \vspace{-5mm}
\end{definition}

\citet{GhoshRS09} showed that the discrete Laplace mechanism is~$\eps$-DP (with or without the clipping step).

\vspace{-2mm}
\section{Preliminaries}

\begin{definition}[Circulant Matrix]
    Let~$p \in \N^+$ and~$\vec{c} = \begin{bmatrix}
                c_0, c_1, \cdots, c_{p - 1} 
    \end{bmatrix} \in \C^p$
    be a row vector.
    Define the (right) circular shift operator~$\sigma: \C^p \rightarrow \C^p$ by
    $
        \sigma \PAREN{\vec{c}} = \begin{bmatrix}
                c_{p - 1}, c_0, \cdots, c_{p - 2} 
        \end{bmatrix}.
    $
    The circulant matrix~$\CirculantMatrix{\vec{c}}$ is a matrix whose rows~$\vec{c}, \sigma \PAREN{\vec{c}}, \ldots, \sigma^{p - 1} \PAREN{\vec{c}}$ are generated by iteratively applying the operator~$\sigma$ on~$\vec{c}$.
\end{definition}

\begin{fact}[Eigenvectors and Eigenvalues~\citep{Gray05}]
    \label{fact: eigenvalues and vectors of circulant matrix}
    Let~$w_t \doteq e^{- \frac{2 \pi t \cdot i}{p} },$ for each~$t \in \IntSet{0}{p - 1}$.
    Then the eigenvalues~$\eigenvalue_0, \ldots, \eigenvalue_{p - 1}$ and the corresponding eigenvectors~$\vec{v}_0, \ldots, \vec{v}_{p - 1}$ of the matrix
    $
        \CirculantMatrix{
            \vec{c}
        }
    $ 
    are given by 
    \vspace{-1mm}
    \begin{align}
        \varphi_t &\doteq c_0 + c_1 (w_t)^1 + \cdots + c_{p - 1} (w_t)^{p - 1}, 
        \\
        \vec{v}_t &\doteq 1 / {\sqrt{p}} \cdot \begin{bmatrix}
        1   ,  (w_t)^1   ,  \cdots  , (w_t)^{p - 1}
        \end{bmatrix}^\top.
        \label{eq: definition of column of fourier transform matrix}
    \end{align}
    \vspace{-1mm}
\end{fact}

\vspace{-4mm}
\begin{fact}[Fast Fourier Transform (FFT)~\citep{CLRS09}]
    \label{fact: fast fourier transform}
    Let~$\DFTMatrix \doteq \bracket{ \vec{v}_0 \cdots \vec{v}_{p - 1} } \in \C^{p \times p}$ be the matrix composed of column vectors
    $\vec{v}_t$ (defined in Equation~\eqref{eq: definition of column of fourier transform matrix})
    for $t = 0, 1, \ldots, p - 1.$
    Then for each~$\vec{x} \in \C^p$, $\DFTMatrix \vec{x}$, $\DFTMatrix^{-1} \vec{x}$,~$\vec{x}^\top \DFTMatrix$ and~$\vec{x}^\top \DFTMatrix^{-1}$ can be computed in~$O(p \log p)$ time.
\end{fact}

\vspace{-2mm}
\section{Near-linear Time Algorithm}
\label{sec: near-linear algorithm}

In this section, we introduce an algorithm that achieves the properties stated in Theorem~\ref{theorem: property of profile approximation algorithm}.
Our presentation is structured as follows.
In Section~\ref{sec: unfolding clipped histogram}, we demonstrate that we can transform the histograms published by the discrete Laplace mechanisms with the clipping step into ones without clipping steps, ensuring they follow the same distribution. 
This allows us to address perturbed histograms with and without clipping in a unified approach. 
Subsequently, in Section~\ref{sec: Searching for Profiles}, we formulate the task of searching for a good estimate of the profile as an optimization problem.
Section~\ref{sec: optimization algorithm} presents a near-linear algorithm for systematically finding an approximate solution to the optimization problem.
Sections~\ref{sec: Time Analysis} and~\ref{sec: error analysis} analyze the time complexity and the approximation error of the proposed algorithms, respectively.
Finally, Section~\ref{subsec: comparison} provides a more detailed comparison of our algorithm with the previous approach.

\vspace{-1mm}
\subsection{Unfolding Clipped Histograms}
\label{sec: unfolding clipped histogram}

Assume that~$\privHist$ has been clipped, so that each entry has value between~$0$ and~$n$. 
Based on the clipped histogram, we can construct a new histogram that has the same distribution as the an unclipped one. 
Consider an entry~$\privHist[\ell] = \clip{\hist[\ell] + \DiscreteLapNoise{e^{-\eps}}}{0}{n}$ for some~$\ell \in \DataDomain$, and the recovery procedure: 
\begin{equation} 
    \begin{array}{c}
        \privHist'[\ell] \doteq
            \begin{cases}
                \privHist[\ell],    & \text{ if } 0 < \privHist[\ell] < n, \\
                \privHist[\ell] - \ExpNoise{e^{-\eps}},    & \text{ if } \privHist[\ell] = 0, \\
                \privHist[\ell] + \ExpNoise{e^{-\eps}},    & \text{ if } n = \privHist[\ell],
            \end{cases}
    \end{array}
\end{equation}
where $\ExpNoise{e^{-\eps}}$ denotes a random variable following the geometric distribution: 
$
    \P{ \ExpNoise{e^{-\eps}} = t } = \PAREN{1 - e^{-\eps}} \cdot e^{-\eps \cdot {t} }, \, \forall t \in \N.
$

\begin{lemma}
    \label{lemma: recover clipped histogram}
    $\privHist'[\ell]$ has the same distribution as~$\hist[\ell] + Z$, where~$Z$ is an independent copy of~$\DiscreteLapNoise{e^{-\eps}}$.
\end{lemma}

Appendix~\ref{appendix: sec: unfolding clipped histogram} contains the detailed proof of the theorem. 
Here, we present a high-level overview of the proof's intuition: when restricting the two-directional discrete Laplace random variable~$\DiscreteLapNoise{e^{-\eps}}$ to a single direction, it reduces to a random variable that follows the geometric distribution. 
It is a well-known fact that for a geometric random variable, truncating it at a certain value and adding an independent copy of the random variable at the truncated value preserves its geometric distribution. 
This property is commonly referred to as the "memorylessness" property of the geometric distribution.

\subsection{Searching for Profiles} 
\label{sec: Searching for Profiles}

A natural attempt would be estimating the profile~$\profile$ with~$\boundedProfileOfNoisyHist\IntSet{0}{n} \doteq \paren{\boundedProfileOfNoisyHist[0], \ldots, \boundedProfileOfNoisyHist[n]}$, where~$
    \boundedProfileOfNoisyHist [t] = 
        \PAREN{1 / d} \cdot | \{\ell \in \DataDomain : \privHist[\ell] = t \} |,
    \forall t \in \Z. 
$ 
There are several issues with such an estimate. First, $\boundedProfileOfNoisyHist\IntSet{0}{n}$ is not an empirical profile, as its entries do not sum up to $1$. 
An even more significant concern is that such an estimate can incur an error of $\Omega(1)$.

\vspace{-3mm}
\paragraph{Example.} 
Consider the example where~$\hist[\ell] = 1, \forall \ell \in \DataDomain.$
Hence, $\profile[1] = 1$ and~$\profile[t] = 0$ for all~$t \neq 1$.
But, it can be shown that~$\mathbb{E}[\boundedProfileOfNoisyHist[1] ] = \frac{1 - e^{-\eps}}{1 + e^{-\eps}}$ and~$\mathbb{V}ar[\boundedProfileOfNoisyHist[1] ] \le \frac{1}{d} \cdot \frac{1 - e^{-\eps}}{1 + e^{-\eps}}.$
By Chebyshev's inequality,
when~$\eps \in \Theta(1)$, with constant probability, $\|\boundedProfileOfNoisyHist\IntSet{0}{n} - \profile \|_p \ge | \boundedProfileOfNoisyHist[1] - \profile[1] | \in \Omega(1),$ for~$p = 1, 2, \infty.$

To deepen our understanding of the relationship between $\boundedProfileOfNoisyHist$ and $\profile$, let us consider a fixed $t \in \IntSet{0}{n}$.
There are $\profile[t] \cdot d$ items in $\DataDomain$ such that $\hist[\ell] = t$. For each of these items, the random variable $\privHist[\ell] = \hist[\ell] + \DiscreteLapNoise{e^{-\eps}}$ is independent and identically distributed.
Our \emph{analysis} considers an algorithm, outlined in Algorithm~\ref{algo: private profile generator}, for generating $\boundedProfileOfNoisyHist$ directly from $\profile$.
We emphasize that our final algorithm is the combination of Algorithms~\ref{algo: fast inversion} and~\ref{algo: rounding}, and does not make use of Algorithm~\ref{algo: private profile generator}.

\vspace{-1mm}
\begin{algorithm}[!ht]
    \caption{Private Profile Generator~$\cA \PAREN{ \estProfile }$}
    \label{algo: private profile generator}
    \hspace{2mm}{\bfseries Input:~$\estProfile$}
    \begin{algorithmic}[1]
        \STATE Initialize an all zero profile~$\vec{v}$
        \STATE {\bf for} $t \in \IntSet{0}{n}$ {\bf do}
            \STATE $\quad$ 
            {\bf for} $j \leftarrow 1$ to~$\estProfile[t] \cdot d$ {\bf do}
                \STATE $\quad$ $\quad$ $I \leftarrow t + \DiscreteLapNoise{e^{-\eps}}$
                \STATE $\quad$ $\quad$ $\vec{v}[I] \leftarrow \vec{v}[I] + {1 / d}$
        \STATE \textbf{return} $\vec{v}$  \hfill\COMMENT{the perturbed profile}
    \end{algorithmic}
\end{algorithm}

To search for a good approximation~$\estProfile$ of~$\profile$, our strategy shifts towards searching for a~$\estProfile$ such that $\cA(\estProfile)$ (the vector returned by Algorithm~\ref{algo: private profile generator} given input~$\estProfile$) is close to~$\boundedProfileOfNoisyHist$.
The search is framed as an optimization problem.
Our goal is then to design efficient algorithm for solving the optimization problem, and to prove that the proximity between $\cA(\estProfile)$ and~$\boundedProfileOfNoisyHist$ indeed implies the proximity between~$\estProfile$ and~$\profile$.

\paragraph{Optimization Problem.}
Let $B \in \R^+$ be the parameter (as defined in Theorem~\ref{theorem: property of profile approximation algorithm}), and $\cE_{B}$ be the event that the all noises introduced by Algorithm~\ref{algo: private profile generator} are absolutely bounded by~$B$, and 
$\Lb \doteq -B$, $\Ub \doteq n + B$, 
$\boundedProfileOfNoisyHist \doteq \paren{\boundedProfileOfNoisyHist[\Lb], \ldots, \boundedProfileOfNoisyHist[\Ub]}.$
Given~$p = 1, 2$ of~$\infty$, the optimization problem we rely on is given by 
{\small
\begin{equation} 
\begin{aligned}
  \min_{\estProfile \in \R^{n + 1}}    && \quad \norm{ \E{\cA }{\cA \PAREN{\estProfile} \mid \cE_B } - \boundedProfileOfNoisyHist }_p&  \\
  \text{s.t.}           && \quad \sum_{t = 0}^n \estProfile[t]  = 1,& \\[1mm]
                        && \quad 0 \le \estProfile[t] \le 1,        &\quad \forall t \in \IntSet{0}{n}.
\end{aligned}
\tag{$\text{P}_1$}\label{opt-original}
\end{equation}
\vspace{-3mm}
}

To simplify the notation, when it is clear from the context, we write~$\E{\cA }{\cA \PAREN{\estProfile} \mid \cE_B }$ as~$\E{\cA \PAREN{\estProfile} \mid \cE_B }$.
In what follows, we provide motivation for this formulation.

{\it Conditioning on~$\cE_B.$}
Without the event~$\cE_B$, $\cA \PAREN{\estProfile}$ can have unbounded support, as the noises introduced by Algorithm~\ref{algo: private profile generator} are unbounded. 
Conditioned on event~$\cE_B$, only~$\cA \PAREN{\estProfile}[\Lb], \ldots, \cA \PAREN{\estProfile}[\Ub]$ process nonzero values. 
Therefore, $\E{\cA \PAREN{\estProfile} \mid \cE_B }$ can be viewed as a vector indexed from~$\Lb$ to~$\Ub$.
Further, by proper choice of~$B$, the event~$\overline{\cE_B}$ only occurs with sufficiently small probability.
Therefore, the event~$\cE_B$ encompasses the majority of cases.

{\it Taking the Expectation.}
Without the expectation, the error~$\| {\cA \PAREN{\estProfile} } - \boundedProfileOfNoisyHist \|_p$ is itself a random variable (even conditioned on~$\cE_B$).
The objective in (\ref{opt-original}) is a relaxation of the expected error, as it provides a lower bound for the latter. 
To see this, note that~$\norm{ \cdot }_p$ is convex for~$p \ge 1$, and according to Jensen's inequality, we have 
\begin{equation*} 
    \hspace{-1mm}
    \begin{array}{c}
        \norm{ \E{\cA}{\cA \PAREN{\estProfile} \mid \cE_B } - \boundedProfileOfNoisyHist }_p
            \le \E{\cA}{\norm{ {\cA \PAREN{\estProfile} } - \boundedProfileOfNoisyHist }_p \mid \cE_B }.
    \end{array}
\end{equation*}

{\it Matrix Form.}
We now reformulate the above in matrix form.
The primary advantage of this formulation lies in its potential for expressing the objective in special matrix form, facilitating the development of efficient optimization algorithms.
Let $\vec{c} \in \R^{n + 2B + 1}$ be a vector defined as follows:
\begin{itemize}[leftmargin=4.5mm, topsep=1pt, itemsep=1pt, partopsep=1pt, parsep=1pt]
    \item The first $B$ entries are set to $1, e^{\text{-}\eps}, \ldots, e^{\text{-}\eps \cdot B}$.
    \item The last $B$ entries are set to $e^{\text{-}\eps \cdot B}, e^{\text{-}\eps \cdot (B - 1)}, \ldots, e^{\text{-}\eps}$.
    \item All other entries are set to $0$.
\end{itemize}
Define 
$
    \ProbNorm 
        \doteq \frac{ 
                    1 + e^{ -\eps } - 2 \,(e^{ -\eps })^{ B + 1 }
                }{ 
                    1 - e^{ -\eps } 
                },
$
and
\begin{equation} 
    \label{equa: def of trans matrix}
    \begin{array}{c}
        \TransMatrix 
            \doteq \frac{1}{\ProbNorm} \, \CirculantMatrix{\vec{c}}. 
    \end{array}
\end{equation}
    
\vspace{-2mm}
To maintain consistency with~$\boundedProfileOfNoisyHist$, we assume that both rows and columns of~$\TransMatrix$ are indexed by~$\IntSet{\Lb}{\Ub}$.
The following lemma states that, we can express the objective in~\eqref{opt-original} in matrix form.

\begin{lemma}
    \label{lemma: expectation in matrix form}
    Let~$\cA(\cdot)$ be the Algorithm~\ref{algo: private profile generator},~$\TransMatrix$ be matrix defined as Equation~\eqref{equa: def of trans matrix}, and~$\estProfile = \PAREN{ \estProfile[0], \ldots, \estProfile[n] } \in \R_{\ge 0}^{n + 1}$ be a vector such that~$\sum_{t \in \IntSet{0}{n}} \estProfile[t] = 1.$
    Assume that we expand the vector $\estProfile$ to one in $\R^{n + 2B + 1}$ by padding $\estProfile[-B] = 
    \cdots = \estProfile[-1] = 0$ to the front and $\estProfile[n + 1] = 
    \cdots = \estProfile[n + B] = 0$ to the back of~$\estProfile$.
    The following claim holds:
    $
        \E{\cA \PAREN{\estProfile} \mid \cE_B } 
            = \TransMatrix \times \estProfile.
    $
\end{lemma}

Based on the above lemma, the optimization problem~\eqref{opt-original} is equivalent to the following one: 
{\small
    \begin{equation} 
    \hspace{-1mm}
    \begin{aligned}
      \min_{\estProfile \in \R^{n + 2B + 1}}    && \norm{ \TransMatrix \times \estProfile - \boundedProfileOfNoisyHist }_p&  \\
      \text{s.t.}           && \sum_{t = 0}^n \estProfile[t]  = 1,& \\[1mm]
                            &&  0 \le \estProfile[t] \le 1,        &\,\, \forall t \in \IntSet{0}{n}, 
                            \\[1mm]
                            &&  \estProfile[t] = 0,                &\,\, \forall t \in \IntSet{\Lb}{\Ub} \setminus \IntSet{0}{n}.
    \end{aligned}
    \tag{$\text{P}_2$}\label{opt-final}
    \end{equation}
    \vspace{-2mm}
}

\vspace{-2mm}
An important advantage lies in the potential efficiency gained from the circulant matrix~$\TransMatrix$, as detailed in Section~\ref{sec: Time Analysis}, enabling~$\TransMatrix^{-1} \vec{x}$ to be computed in~$O(n \log n)$ time for each vector~$\vec{x}$. 
This motivates the choice of computing~$\TransMatrix^{-1} \boundedProfileOfNoisyHist$ as the approximate profile~$\estProfile$, resulting in zero error in the objective function:
\begin{equation*}
    \begin{array}{c}
        \norm{
            \TransMatrix \times \TransMatrix^{-1} \times \boundedProfileOfNoisyHist
             - \boundedProfileOfNoisyHist
        }_p = 0.
    \end{array}
\end{equation*}
However, it's crucial to note that such a solution may potentially violate other constraints outlined in~\eqref{opt-final}.
This observation motivates the development of an algorithm for systematically solving~\eqref{opt-final}, which will be discussed in the next section.

\subsection{Optimization Algorithms}
\label{sec: optimization algorithm}

In this section, we present effective algorithms to address the optimization problem defined by~\eqref{opt-final}. 
Our approach consists of two steps: initially, we concentrate on a relaxed version of~\eqref{opt-final} and introduce an algorithm that optimally solves it. Subsequently, building upon the solution obtained in the first step, we introduce a rounding procedure designed to yield an approximate solution for~\eqref{opt-final}, satisfying all the specified constraints.
The proofs of all lemmas in this section are included in the 
\emph{Appendix~\ref{appendix: sec: optimization algorithm}.}

\vspace{-3mm}
\paragraph{Relaxation.}
There are two categories of constraints on $\estProfile$ in optimization problem~\eqref{opt-final}: its coordinates should sum up to $ 1 $; and each of its coordinates should be within a given range.
Optimizing~\eqref{opt-final} simultaneously with two constraints is non-trivial, so we initially focus on the first constraint, leading to optimization problem~\eqref{opt-relaxed}:
\begin{equation} 
    \begin{array}{ccccc}
        &\min           &   \norm{ \TransMatrix \times \estProfile - \boundedProfileOfNoisyHist }_p 
        &\text{s.t.}    &   \sum_{t = 0}^n \estProfile[t] = 1.          
    \end{array}
    \tag{$\text{P}_3$}\label{opt-relaxed}
\end{equation}
The optimization problem~\eqref{opt-relaxed} has a simple geometric interpretation.
Since $\TransMatrix$ is invertible (as will be shown in Lemma~\ref{lemma: property of transmatrix}), the function $\Vert \cdot \Vert_{\TransMatrix, p} : \R^{\IntSet{0}{n}} \rightarrow \R$, defined by 
\vspace{-1mm}
$$
\begin{array}{c}
    \Vert \vec{v} \Vert_{\TransMatrix, p} \doteq \Vert \TransMatrix \times \vec{v} \Vert_p, 
    \quad \forall \vec{v} \in \R^{\IntSet{0}{n}}.
\end{array}
\vspace{-1mm}
$$
constitutes a norm.
The objective of~\eqref{opt-relaxed} is equivalent to 
$$
    \begin{array}{c}
        \Vert \TransMatrix \times \paren{ \estProfile - \TransMatrix^{-1} \times \boundedProfileOfNoisyHist }\Vert_p
    \end{array}
$$
Hence,~\eqref{opt-relaxed} seeks to find the point nearest to $\TransMatrix^{-1} \times \boundedProfileOfNoisyHist$ (measured by the distance induced by the norm $\Vert \cdot \Vert_{\TransMatrix, p}$) within the hyperplane specified by $\sum_{t = 0}^n \estProfile[t] = 1$.

This inspires Algorithm~\ref{algo: fast inversion} for solving~\eqref{opt-relaxed}. 
First, it computes a vector~$\vec{u} = {\TransMatrix}^{-1} \times { \boundedProfileOfNoisyHist }$.
Let~$\vecTargetSupport$ be a binary vector such that~$\vecTargetSupport[t] = 1$ only when~$t \in \IntSet{0}{n}$.
If~$\tinner{\vecTargetSupport}{\vec{u}} \neq 1$, the algorithm rectifies this by adjusting~$\vec{u}$ along the direction~$\TransMatrix^{-1} \times \vec{a}$ (by an appropriate amount), where~$\vec{a}$ is the unit vector (in terms of the norm~$\norm{\cdot}_p$) which maximizes~$\tinner{\vecTargetSupport}{\TransMatrix^{-1} \times \vec{a}}$.
In essence, adjusting~$\vec{u}$ along~$\TransMatrix^{-1} \times \vec{a}$ is the most cost-effective way to correct the violation~$\tinner{\vecTargetSupport}{\vec{u}} \neq 1$.
Formally, Lemma~\ref{lemma: correctness of algo fast inversion} holds.

\begin{algorithm}[!t]
    \caption{Fast Inversion~$\algoFstInv$}
    \label{algo: fast inversion}
    
    \hspace{2mm}{\bfseries Input:
    $\privHist$
    and~$p \in\{ 1, 2, \infty\}$}
    
    \begin{algorithmic}[1]

        \STATE Compute~$\boundedProfileOfNoisyHist$ based on~$\privHist$
        \STATE $\vec{u} \leftarrow {\TransMatrix}^{-1} \times { \boundedProfileOfNoisyHist }$
        
        \vspace{-3mm}
        \STATE 
        Let~$\vecTargetSupport$ be a vector, s.t., 
        $
            \vecTargetSupport[t] = \begin{cases}
                1, & \forall t \in \IntSet{0}{n} \\
                0, & \text{ otherwise } 
            \end{cases}
        $

        \STATE Let $\vec{a} \leftarrow \arg\max_{\norm{\vec{a}'}_p = 1} 
            \langle{\vecTargetSupport}{\TransMatrix^{-1} \times \vec{a}'}\rangle$
        \STATE $\estProfile \leftarrow \vec{u} -  
                \frac{
                    { \inner{\vecTargetSupport}{\vec{u}} - 1 }
                }{ 
                    \inner{\vecTargetSupport}{\TransMatrix^{-1} \times \vec{a}}
                }  
                \cdot 
                \TransMatrix^{-1} \times \vec{a}
        $ 
        \STATE \textbf{return} $\estProfile$ \hfill\COMMENT{the estimated profile}
    \end{algorithmic}
\end{algorithm}

\begin{lemma}
    \label{lemma: correctness of algo fast inversion}
    Algorithm~\ref{algo: fast inversion} returns an optimal solution for the optimization problem~\eqref{opt-relaxed}. 
\end{lemma}

\paragraph{Rounding.}
Given a solution~$\estProfile$ returned by Algorithm~\ref{algo: fast inversion}, Algorithm~\ref{algo: rounding} post-processes it to ensure satisfaction of all constraints in~\eqref{opt-final}.

Algorithm~\ref{algo: rounding} unfolds in three phases: Phase 1 (Line~\ref{line: round to zero}), wherein the algorithm sets~$\estProfile[t]$ to~$0$ for each~$t \in \IntSet{\Lb}{\Ub} \setminus \IntSet{0}{n}$; Phase 2 (Line~\ref{line: large set} to Line~\ref{line: round from below}), where the algorithm clips~$\estProfile[t]$ to the range of~$[0, 1]$; and Phase 3 (Line~\ref{line: compute adjust amount} to Line~\ref{line: end of fix}), during which the algorithm adjusts~$\estProfile[t]$ to ensure their sum equals~$1$.
It is essential to observe that Phase 3 will not generate any coordinate outside the range~$[0, 1]$. 
Moreover, in Phase 3, it is adequate to solely consider decreasing the coordinate values of~$\estProfile$,
which will be formally explained in the proof of Lemma~\ref{lemma : correctness of algo rounding}.

\begin{algorithm}[!ht]
    \caption{Rounding~$\algoRounding$}
    \label{algo: rounding}
    
    \hspace{2mm}{\bfseries Input:~$\estProfile$}
    \begin{algorithmic}[1]
        
        \STATE $\estProfile[t] \leftarrow 0, \forall t \in \IntSet{\Lb}{\Ub} \setminus \IntSet{0}{n}$ 
        \alglinelabel{line: round to zero}
        
        \STATE $\largeSet \leftarrow \set{t \in \IntSet{0}{n} : \estProfile[t] > 1}$
        \alglinelabel{line: large set}
        \STATE $\smallSet \leftarrow \set{t \in \IntSet{0}{n} : \estProfile[t] < 0}$

        \STATE $\surplus \leftarrow \sum_{t \in \largeSet} \PAREN{1 - \estProfile[t]}$
        \STATE $\deficit \leftarrow \sum_{t \in \smallSet} \PAREN{0 - \estProfile[t]}$

        \STATE $\estProfile[t] \leftarrow 1$, $\forall t \in \largeSet$    
        \alglinelabel{line: round from above}
        \STATE $\estProfile[t] \leftarrow 0$, $\forall t \in \smallSet$
        \alglinelabel{line: round from below}
        \STATE $s \leftarrow \surplus + \deficit$. \hfill {$\triangleright$ we must have~$s \ge 0$}
        \alglinelabel{line: compute adjust amount}
        \IF{$s > 0$} 
            \STATE \hspace{-3mm} Let~$\tau \ge 0$ be the solution for~$\sum_{t = 0}^n \min \set{ \tau, \estProfile[t] } = s$
            \alglinelabel{line: search for the threshold}
            \STATE \hspace{-3mm} $\estProfile[t] \leftarrow \estProfile[t] - \min \set{ \tau, \estProfile[t] }$, $\forall t \in \IntSet{0}{n}$  
            \alglinelabel{line: adjust based on computed threshold}
        \ENDIF
        \alglinelabel{line: end of fix}

        \STATE \textbf{return} $\estProfile$ \hfill \COMMENT{the rounded profile}
    \end{algorithmic}
\end{algorithm}

\begin{lemma}
    \label{lemma : correctness of algo rounding}
    The vector returned by Algorithm~\ref{algo: rounding} satisfies all constraints in the optimization problem~\eqref{opt-final}. 
\end{lemma}

\subsection{Time Analysis} 
\label{sec: Time Analysis}

In this section, we study the running time of both Algorithm~\ref{algo: fast inversion} (Theorem~\ref{theorem: running time of algo fast inversion}) and Algorithm~\ref{algo: rounding} (Theorem~\ref{theorem: running time of algo rounding}).

\begin{theorem}
    \label{theorem: running time of algo fast inversion}
    Algorithm~\ref{algo: fast inversion} runs in $O(d + n \log n)$ time.
\end{theorem}

The proof of Theorem~\ref{theorem: running time of algo fast inversion} relies heavily on the properties of~$\TransMatrix$ as a circulant matrix.

\begin{lemma}
    \label{lemma: property of transmatrix}
    The eigenvalues~$\eigenvalue_t, t = -B, \ldots, n + B$ of $\TransMatrix$ can be computed in~$O(n)$ time.
    Further, $\TransMatrix$ is invertible and 
    \begin{equation*} 
        \begin{array}{c}
            \TransMatrix^{-1}
                = \DFTMatrix \times 
                \diag{\eigenvalue_{-B}^{-1}, \eigenvalue_{-B + 1}^{-1}, \ldots, \eigenvalue_{n + B}^{-1} }
                \times \DFTMatrix^{-1},
        \end{array}
    \end{equation*}
    where~$\diag{\eigenvalue_{-B}^{-1}, \eigenvalue_{-B + 1}^{-1}, \ldots, \eigenvalue_{n + B}^{-1} }$ is a diagonal matrix, and~$\DFTMatrix \doteq \bracket{ \vec{v}_{-B} \cdots \vec{v}_{n + B} }$ is the matrix comprising column vectors
    $
        \vec{v}_t = \frac{1}{\sqrt{n + 2B + 1}} \cdot \begin{bmatrix}
            1   ,  (w_t)^1   ,  \cdots  , (w_t)^{n + 2B}
        \end{bmatrix}^\top, \,
    $
    $   w_t \doteq
            e^{- \frac{2 \pi \cdot \PAREN{t + B} \cdot i}{n + 2B + 1} },
    $
    for~$t = -B, \ldots, n + B$.
\end{lemma}

The proof of Lemma~\ref{lemma: property of transmatrix} is included in the Appendix~\ref{appendix: sec: time analysis}.
Now we are ready to prove Theorem~\ref{theorem: running time of algo fast inversion}. 

\begin{proof}[Proof of Theorem~\ref{theorem: running time of algo fast inversion}]
    First, computing the empirical profile~$\boundedProfileOfNoisyHist$ from $\privHist$ takes $O(d)$ time.
    Next, by fast Fourier transform (Fact~\ref{fact: fast fourier transform}), multiplying a vector with both $\DFTMatrix$ and $\DFTMatrix^{-1}$ can be performed in~$O((n + 2B + 1) \log (n + 2B + 1)) = O(n \log n)$ time.
    Based on the decomposition of~$\TransMatrix^{-1}$ (Lemma~\ref{lemma: property of transmatrix}), we see that multiplying~$\TransMatrix^{-1}$ with a vector takes~$O(n \log n)$ time.
    Therefore, it suffices to show that~$\vec{a}$ can also be found in~$O(n \log n)$ time. 

    For consistency, assume that the coordinates of~$\vec{a}$ are indexed from~$-B$ to~$n + B$ and write~$\vec{a} = \bracket{a_{-B}, \ldots, a_{n + B}}$.
    Then 
    $
        \tinner{\vecTargetSupport}{\TransMatrix^{-1} \times \vec{a}}
            = \sum_{t = -B}^{n + B} \tinner{\vecTargetSupport}{\TransMatrix^{-1}[\cdot, t]} \cdot a_t,
    $
    where~$\TransMatrix^{-1}[\cdot, t]$ is the~$t^{(th)}$ column of~$\TransMatrix^{-1}$.
    For convenience, let~$c_t$ be a short hand for~$\tinner{\vecTargetSupport}{\TransMatrix^{-1}[\cdot, t]}$.
    Observe that
    $
        \vec{c} \doteq \bracket{c_{-B}, \ldots, c_{n + B}} = \PAREN{\vecTargetSupport}^\top \times \TransMatrix^{-1}
    $
    can be computed in $O(n \log n)$ time.
    Once the~$c_t$ are known, finding~$\vec{a}$ is equivalent to solving the following optimization problem: 
    $
        \max_{\norm{\vec{a}}_p = 1} \inner{\vec{c}}{\vec{a}}.
    $
    
    It is straightforward to see that: 
    i) $\vec{a} = \text{sign}(c_t) \cdot \vec{e}_t$ where~$\vec{e}_t$ is a standard basis vector and~$t = \arg\max_t \card{c_t}$ when~$p = 1$; 
    ii) $\vec{a} = \frac{\vec{c}}{\norm{\vec{c}}_2}$ when~$p = 2$; 
    iii) $\vec{a} = \bracket{\text{sign}\PAREN{c_{-B}}, \ldots, \text{sign}\PAREN{c_{n + B}}}^\top$ when~$p = \infty$. 
    \vspace{-4mm}

    \vspace{-2mm}
\end{proof}

\vspace{-2mm}
\begin{theorem}
    \label{theorem: running time of algo rounding}
    Algorithm~\ref{algo: rounding} has an implementation which runs in~$O(n \log n)$ time.
\end{theorem}

The detailed proof of the theorem is included in the Appendix~\ref{appendix: sec: time analysis}.
At a high level, we can initially sort the $\estProfile[t]$ in non-decreasing order in $O(n \log n)$ time. Subsequently, we identify the smallest integer $t^*$ for which $\estProfile[t^*] \ge \tau$. It can be demonstrated that $\tau = \frac{s - \sum_{t = 0}^{t^* - 1} \estProfile[t]}{n + 1 - t^*}$.

\subsection{Error Analysis}

\label{sec: error analysis}

In this section, we analyze the estimation error when we combined algorithms~\ref{algo: fast inversion} and~\ref{algo: rounding}, and show that this results in the estimation error stated in Theorem~\ref{theorem: property of profile approximation algorithm}.

The approximation errors arise from two sources: 
the difference between the optimal solution of~\eqref{opt-relaxed} and the actual profile; and the error introduced during the rounding procedure.
We will demonstrate that the errors resulting from the former are of the same order of magnitudes as those indicated in Theorem~\ref{theorem: property of profile approximation algorithm}, while the latter introduces only additional constants to the error of the former.
Formally, the following two lemmas hold. 

\begin{lemma}[{Error from Optimization~\eqref{opt-relaxed}}]
    \label{lemma: error from opt-relaxed}
    Given the profile~$\boundedProfileOfNoisyHist$ of a private version~$\privHist$ of a histogram~$\hist$ published by the~$\eps$-DP discrete Laplace mechanism and a parameter~$p = 1, 2$ or~$\infty$, with probability at least $1 - \eta$, Algorithm~\ref{algo: fast inversion} returns a vector~$\estProfile$ with the estimation error~$\err{p} \doteq \| \estProfile - \profile \|_p$ satisfying
    Inequalities~(\ref{ineq: error 1},\ref{ineq: error 2},\ref{ineq: error 3}).
\end{lemma}

\begin{lemma}[Error from Rounding]
    \label{lemma: error of rounding}
    Given a vector~$\estProfile$ returned by Algorithm~\ref{algo: fast inversion}, let~$\err{p}$ denote the error before running Algorithm~\ref{algo: rounding}, and~$\err{p}'$ denote the error afterwards.
    Then~$
        \err{p}' \le \err{p} \,
    $ 
    for $p = 1, 2,$ 
    and
    $
        \err{\infty}' \le 2 \, \err{\infty}.
    $
\end{lemma}

Combing Lemma~\ref{lemma: error from opt-relaxed} and~\ref{lemma: error of rounding} immediately proves the error bounds in Theorem~\ref{theorem: property of profile approximation algorithm}.
In the rest of the section, we present the proofs for Lemma~\ref{lemma: error from opt-relaxed} and Lemma~\ref{lemma: error of rounding} separately.

\vspace{-2mm}
\paragraph{Error from Optimization~\eqref{opt-relaxed}}

\begin{proof}[Proof of Lemma~\ref{lemma: error from opt-relaxed}]
    \vspace{-3mm}
    We need to prove that, the closeness between~$\TransMatrix \times \estProfile$ and~$\boundedProfileOfNoisyHist$ implies the closeness between~$\estProfile$ and~$\profile$.
    To bound $\| \estProfile - \profile  \|_p,$ we observe that  

    \vspace{-3mm}
    \resizebox{\linewidth}{!}{
      \begin{minipage}{1.2\linewidth}
        \begin{align*}
             \norm{ \estProfile - \profile }_p
            &= 
            \norm{ \TransMatrix^{-1} \times \TransMatrix \times \PAREN{ \estProfile - \profile } }_p 
            \le \norm{ \TransMatrix^{-1} }_p \cdot \norm{ \TransMatrix \times \PAREN{ \estProfile - \profile } }_p \\
            &\le \norm{ \TransMatrix^{-1} }_p \cdot \PAREN{ 
                \norm{ \TransMatrix \times \estProfile - \tilde{f} }_p 
                + \norm{ \tilde{f} - \TransMatrix \times \profile }_p 
            } \\
            &\le \norm{ \TransMatrix^{-1} }_p \cdot \PAREN{ 2 \cdot \norm{ \tilde{f} - \TransMatrix \times \profile }_p },
            \label{ineq: relate fianl error with immediate one}
        \end{align*}
      \end{minipage}
    }

    where the first inequality comes from the definition of matrix norm $\| \TransMatrix^{-1} \|_p \doteq \max_{\vec{v}  : \| \vec{v} \|_p = 1 } \| \TransMatrix^{-1} \times \vec{v} \|_p$; the last one from that~$\estProfile$ is the solution for optimization problem~\eqref{opt-relaxed}.
    It remains to bound~$\|\TransMatrix^{-1} \|_p$ and~$\| \tilde{f} - \TransMatrix \times \profile \|_p$ separately.

    \begin{lemma} Let~$\eigenvalue_t, t = -B, \ldots, n + B$ be the eigenvalues of~$\TransMatrix$. 
    Then 
        \label{lemma: bounds on matrix norms}

        \vspace{-4mm}
        \resizebox{\linewidth}{!}{
          \begin{minipage}{1.15\linewidth}
            \begin{align*}
                \norm{ \TransMatrix^{-1} }_1 
                &= \norm{ \TransMatrix^{-1} }_\infty
                \le \frac{
                        2 + e^{ -\eps } + e^{ \eps }
                    }{
                        e^{ \eps } - e^{ -\eps } - 4 \cdot e^{-\eps \cdot B }
                    } 
                    \cdot \ProbNorm, 
                \\
                \norm{ \TransMatrix^{-1} }_2 
                &\le 
                \max_t \frac{1}{|\eigenvalue_t|}
                \le \frac{
                        \ProbNorm \cdot \PAREN{ 1 + e^{ -\eps } }
                    }{
                        \PAREN{
                            1 - e^{ -\eps } - 2 \cdot e^{ -\eps \cdot \PAREN{ B + 1 } } 
                        } 
                    }.
            \end{align*}
          \end{minipage}
        }
    \end{lemma}

    \begin{lemma}
        \label{lemma: deviation of profile of noisy histogram}
        With probability at least $1 - \eta$, it holds that 
        
        \vspace{-3mm}
        \resizebox{\linewidth}{!}{
          \begin{minipage}{1.18\linewidth}
            \begin{align*}
                \norm{ \boundedProfileOfNoisyHist - \TransMatrix \times \profile }_1 
                    &\le \frac{1}{ \sqrt{d} } \cdot { 
                        \sum_{t = -B}^{n + B} \sqrt{ \E{ \boundedProfileOfNoisyHist [t] } }
                    }
                    + \sqrt{
                        \frac{2 \cdot \ln \frac{1}{\eta}}{d}
                    }. \\
                \norm{ \boundedProfileOfNoisyHist - \TransMatrix \times \profile }_2 
                    &\le \sqrt{\frac{1}{d}} + \sqrt{
                        \frac{\ln \frac{1}{\eta}}{d}
                    }. \\
                \norm{ \boundedProfileOfNoisyHist - \TransMatrix \times \profile }_\infty 
                    &\le \frac{1}{ \sqrt{d} } \cdot \sqrt{ 2 \cdot \frac{1}{\ProbNorm} \cdot \ln \frac{n}{\eta} } + \frac{1}{ 3 d } \cdot \ln \frac{n}{\eta} .
            \end{align*}
          \end{minipage}
        }
    \end{lemma}
    The proof of Lemma~\ref{lemma: bounds on matrix norms} and Lemma~\ref{lemma: deviation of profile of noisy histogram} are included in the Appendix~\ref{appendix: sec: error analysis}. 
    Finally, combining these lemmas, 
    the assumption of~$B$ in Theorem~\ref{theorem: property of profile approximation algorithm},
    and that 
    $
        \ProbNorm 
        = \frac{ 
            1 + e^{ -\eps } - 2 \cdot (e^{ -\eps })^{ B + 1 }
        }{ 
            1 - e^{ -\eps } 
        }
    $
    finishes the proof.
\end{proof}

\paragraph{Error from Rounding.}
The full proof for Lemma~\ref{lemma: error of rounding} is provided in the Appendix~\ref{appendix: sec: error analysis}. 
Here is a brief outline of the proof for the case of~$p = 1$ to offer some insight. 
The proof for the cases of~$p = 2, \infty$ are similar but much more technical.
The algorithm is divided into three phases:
\begin{itemize}[leftmargin=4.5mm, topsep=1pt, itemsep=1pt, partopsep=1pt, parsep=1pt]
    \item Phrase 1 (Line~\ref{line: round to zero}): In this phase, the algorithm rounds~$\estProfile[t]$ to~$0$ for each~$t \in \IntSet{-B}{n + B} \setminus \IntSet{0}{n}$.

    \item Phase 2 (Line~\ref{line: large set} to Line~\ref{line: round from below}): In this phase, the algorithm clips the~$\estProfile[t]$ to the range of~$[0, 1]$.

    \item Phase 3 (Line~\ref{line: compute adjust amount} to Line~\ref{line: end of fix}): In this phase, the algorithm adjusts the~$\estProfile[t]$ so that they sum up to~$1$.
\end{itemize}
It is evident that the quantity~$\card{\estProfile[t] - \profile[t]}$ only decreases in Phases 1 and 2, as~$\profile[t] \in [0, 1]$. 
Further it can be shown that in Phase 2,~$\| \estProfile - \profile \|_1$ decreases by an amount of~$-\surplus + \deficit$, and in Phase 3,~$\| \estProfile - \profile \|_1$ can increase by at most~$\surplus + \deficit$. Hence overall,~$\| \estProfile - \profile \|_1$ decreases.

\subsection{Comparison with Previous Approach}
\label{subsec: comparison}

\citet*{DworkNPRY10} also studied the profile estimation (referred to as the \emph{$t$-incidence}) problem, with a key difference from our work: their approach aimed for pan-privacy, which resulted in significantly more noise added to the histogram \(\vec{h} \in [0 .. n]^d\). 
To simplify the comparison, we assume a sampling step involved in their method has already been executed, and focus on the full histogram of sampled items. 
Let \(\tilde{h}\) denote the noisy histogram. 
For each \(\ell \in [d]\), their requirement was that
\[
\big(
    \max_{t \in [0 .. n]} \Pr[\tilde{h}[\ell] = t]
\big) \big/ \big(
    \min_{t \in [0 .. n]} \Pr[\tilde{h}[\ell] = t]
\big)
\le e^\eps.
\]
Specifically, they used noise scaling with \(n\): \(\tilde{h}[\ell] \leftarrow \vec{h}[\ell] + Z \pmod{n + 1}\), where \(Z\) is a noise variable such that \(\Pr[Z = t] \propto e^{- \eps \cdot t / (n + 1)}\). 

In contrast, our paper considers adding symmetric discrete Laplace noise \(\mathbb{DL}ap(e^{-\eps})\) to each entry, which is fundamental and widely used for publishing privatized histograms. Our technique is applicable even in the pan-private setting, with appropriately scaled noise.

Another minor difference lies in the optimization problem: their linear program (LP) aims to achieve solutions with \(\ell_\infty\) error, whereas our algorithm supports \(\ell_1/\ell_2/\ell_\infty\) errors. Moreover, they did not provide an almost linear-time algorithm for solving their LP. However, the almost linear-time algorithm developed in our paper can solve their LP as well, as it also involves a circulant matrix.

\vspace{-2mm}
\section{Lower Bound}
\label{sec: lower bound}

In this section, we discuss the lower bound in Theorem~\ref{theorem: accuracy lower bound of profile estimation}. 
We will begin by providing the formal definition of \emph{updatable} histogram publishing algorithms and the \emph{sensitivity} of the profile reconstruction algorithms when composed with the private histogram publishing algorithm.
To establish the lower bound, we will then present a reduction from the inner product estimation problem in the two-party setting.

Our lower bound applies to an algorithm $\mathcal{A}$ that is updatable, meaning that there exists an update algorithm~$\cU$, s.t., given the output $\mathcal{A}(\hist)$ and some other histogram $\vec{h}'$, $\cU$ can adjust the output directly (by computing $\cU(\cA(\hist), \privHist)$), as if the input to $\mathcal{A}$ is $\vec{h} + \vec{h}'.$
Formally, 

\begin{definition}[Updatable Algorithm]
    A (randomized) algorithm~$\cA : \IntSet{0}{n}^d \rightarrow \cY$ is called updatable, if there exists another (randomized) algorithm~$\cU : \cY \times \IntSet{0}{n}^d \rightarrow \cY$, the \emph{update algorithm}, such that for every~$\hist, \hist' \in \IntSet{0}{n}^d$, $\cU(\cA(\hist), \hist')$ has the same distribution as~$\cA(\hist + \hist')$.
\end{definition}

\begin{example}
    If $\cA$ is the discrete Laplace mechanism, then it is updatable.
    Let~$\cU$ be the vector addition operation:
    $
        \cU(\cA(\hist), \hist')
        = \hist + \DiscreteLapNoise{e^{-\eps}}^d + \hist'
        = \cA(\hist +  \hist'), 
    $
    where~${\DiscreteLapNoise{e^{-\eps}}}^d \in \Z^d$ is a vector composed of independent discrete Laplace random variables. 
\end{example}

When~$\cA$ represents a randomized algorithm, it can be conceptualized as a function, denoted as~$F_A$, acting on both the input of~$\cA$ and the randomness inherent in its execution (commonly referred to as ''coin tosses'' in the literature). 
This randomness is formally represented by the random variable~$C_A$.
In what follows, denote~$\cR: \cY \rightarrow [0, 1]^{n + 1}$ as an algorithm trying to reconstruct the profile of~$\hist$ when given~$\cA ( \hist ).$

\begin{definition}
    \label{def: composed sensitivity}
    The sensitivity of~$\cR \circ \cA$ is defined as the maximum deviation of~$\cR \circ F_A$, over all possible realization~$c_A$ of~$C_A$ and all possible pair of neighboring inputs of~$\cA:$

    \vspace{-3mm}
    \resizebox{\linewidth}{!}{
      \begin{minipage}{1.08\linewidth}
        \begin{equation*}
            \Delta_{\cR \circ \cA} \doteq \max_{c_A} \max_{\hist \sim \hist'} 
            \norm{
                \cR \PAREN{ F_A \PAREN{ \hist, c_A } } 
                - \cR \PAREN{ F_A \PAREN{ \hist', c_A } }
            }_\infty, 
        \end{equation*}
      \end{minipage}
    }
    where the second maximum is over all~$\hist, \hist' \in \IntSet{0}{n}^d$ such that~$\| \hist - \hist' \|_1 \le 1 $
\end{definition}

In Appendix~\ref{appendix: sec: lower bound}, we show that when \(\cA\) is the discrete Laplace mechanism, \(\cR\) is the profile reconstruction algorithm stated in Theorem~\ref{theorem: property of profile approximation algorithm}, and \(\frac{1}{\sqrt{d}} \in o \bigparen{ \big( \frac{e^\eps - 1}{e^\eps + 1} \big)^2 \cdot \frac{\eps}{e^{c \cdot \eps} }}\), \(\Delta_{\cR \circ \cA}\) satisfies \(\Delta_{\cR \circ \cA} \in o \big( \frac{1}{\sqrt{d}} \cdot \frac{\eps}{e^{c \cdot \eps}} \big)\). Hence, the lower bound in Theorem~\ref{theorem: accuracy lower bound of profile estimation} applies to our algorithms.

\paragraph{Two Party Differential Privacy.}
We will prove the lower bound by establishing a connection between the profile estimation problem and the inner product estimation problem in the two party differential privacy setting~\citep{McGregorMPRTV10}. 

Two players, Alice and Bob, equipped with unrestricted computational capabilities, collaborate to compute the inner product $\inner{\vec{x}}{\vec{y}}$, 
where $\vec{x} \in \set{-1, 1}^d$ is the binary string held by Alice, and $\vec{y} \in \set{-1, 1}^d$ is one held by Bob.
The collaboration follows a pre-agreed protocol, denoted as $\Pi$.
In this protocol, they alternate turns to communicate with each other. 
During each player's turn, the communicated information is determined by their own vector, the cumulative communication up to that point, and possibly their randomness.
The sequences of messages, which collectively form the protocol's transcript, are denoted as $\Pi$.
To protect the data of both parties, it is required that each party's view of $\Pi$, as a function of the other party's data, should be differentially private. 

Formally, the view of Alice, 
denoted as~$\viewA : \set{-1, 1}^d \rightarrow \cZ$ 
(where~$\cZ$ is the set of all possible sequences of messages), is regarded as a (randomized) algorithm with input string~$\vec{y}$ when~$\vec{x}$ is fixed.
The view~$\viewA$ is~$\eps$-DP, if for each measurable~$Z \subseteq \cZ$ and every pair of neighboring vector~$\vec{y}$ and~$\vec{y}'$ s.t.~$\norm{ \vec{y} - \vec{y}' } = 1$, it holds that
\begin{equation*}
    \begin{array}{c}
        \P{ \viewA (\vec{y}) \in Z } 
            \le e^{\eps} \cdot \P{ \viewA (\vec{y}') \in Z }. 
    \end{array}
\end{equation*}
We can define the view of Bob 
$\viewB: \set{-1, 1}^d \rightarrow \cZ$ 
and the~$\eps$-DP property of~$\viewB$ in a similar manner. 
The protocol~$\Pi$ is called~$\eps$-DP, if~$\viewA$ is~$\eps$-DP for all possible~$\vec{x}$, and~$\viewB$ is~$\eps$-DP for all possible~$\vec{y}$.
\citet{McGregorMPRTV10} established a lower bound on inner product estimation, which was further refined by \citet{HaitnerMST22}.

\vspace{1mm}
\begin{fact}[The Two Party Lower Bound~\citep{HaitnerMST22}]
    \label{fact: two party accuracy lower bound}
    An $\eps$-DP two-party protocol that, for some~$\ell \ge \log d$, estimate the inner product over $\set{\text{-1}, 1}^d \times \set{\text{-1}, 1}^d$ up to an additive error~$\ell$ with probability $c \cdot e^{c \cdot \epsilon} \cdot \ell / \sqrt{d}$ (for some universal constant $c > 0$), can be used to construct a key-agreement protocol.
\end{fact}

\vspace{-1mm}
Since an information theoretically secure key agreement does not exist~\citep{HaitnerMST22}, Fact~\ref{fact: two party accuracy lower bound} implies that such protocols do not exist either.

\vspace{-3mm}
\begin{proof}[Sketch Proof of Theorem~\ref{theorem: accuracy lower bound of profile estimation}.]
Based on Fact~\ref{fact: two party accuracy lower bound}, we present a sketch proof for the theorem.
The complete proof is included in the Appendix~\ref{appendix: sec: lower bound}.

We will prove the theorem by contraction. 
Assume there exists an $\eps$-DP and updatable algorithm~$\cA : \IntSet{0}{n}^d \rightarrow \cY$, profile estimation algorithm~$\cR : \cY \rightarrow [0, 1]^{n + 1}$, such that, for every input~$\hist \in \IntSet{0}{n}^d$, 
\begin{equation*}
    \begin{array}{c}
        \mathbb{E} [ \Vert \profile - \cR \paren{ \cA(\hist) } \Vert_\infty ] 
            \in o \bigparen{ {1} /  \paren{\sqrt{d} \cdot e^{c \cdot \eps}} }.
    \end{array}
\end{equation*}

We will construct an~$\eps$-DP protocol~$P$ such that for every~$\vec{x}$ and~$\vec{y}$,
Bob's output (denoted by $m_B$) satisfies 
\begin{equation}
    \label{eq: accuracy contradiction}
    \begin{array}{c}
        \P{ \card{m_B - \inner{\vec{x}}{\vec{y}}} 
            \ge \err{} 
        } 
        \in o(1), 
    \end{array}
\end{equation}
where $\err{} = \sqrt{d} / \paren{2 \cdot c \cdot e^{c \cdot \eps} }$,
contradicting Fact~\ref{fact: two party accuracy lower bound}.

The protocol is outlined by Algorithm~\ref{algo: two party protocol},
where~$\LapNoise{1 / \eps}$ denotes a random variable following the Laplace distribution with probability density~$p(z) = \frac{\eps}{2} \cdot \exp \PAREN{ - \eps \cdot \card{z} }, \forall z \in \R.$
    
\setlength{\textfloatsep}{2pt}%
\begin{algorithm}[!t]
    \caption{Protocol~$P$}
    \label{algo: online sampling algorithm}
    \label{algo: two party protocol}
    {\it Alice}: {\bf send} $m_A \leftarrow \cA(\vec{x} + \vecAllOne)$ to Bob.
    
    \vspace{1mm}
    {\it Bob}: $\circ$ on receiving~$m_A$, compute~$\estProfile \leftarrow \cR \paren{ \cU \paren{ m_A, \vec{y} + \vecAllOne) } }$
    
    \hspace{6.7mm} $\circ$ {\bf publish} $d \PAREN{ \estProfile[4] + \estProfile[0] - \estProfile[2] + 3\Delta_{\cR \circ \cA} \cdot \LapNoise{1 / \eps} }$
\end{algorithm}

Since $\vec{x}, \vec{y} \in \set{\text{-}1, 1}^d$ and $\vec{x} + \vecAllOne, \vec{y} + \vecAllOne \in \set{0, 2}^d$, 
it can be verified that 
\[
    \hspace{-2mm}
    \begin{array}{rl}
        \inner{\vec{x}}{\vec{y}} \hspace{-3mm}
            &= \card{ \set{ \ell \in [d] : \vec{x} [\ell] = \vec{y} [\ell]} } 
            - \card{ \set{ \ell \in [d] : \vec{x} [\ell] \neq \vec{y} [\ell]} } \\
            &= \card{ \set{ \ell \in [d] : \paren{\vec{x} + \vecAllOne} [\ell] = \paren{\vec{y} + \vecAllOne} [\ell] } } \\
            &\, - \card{ \set{ \ell \in [d] : \paren{\vec{x} + \vecAllOne} [\ell] \neq \paren{\vec{y} + \vecAllOne} [\ell] } } \\
            &= d \cdot \paren{\profile[4] + \profile[0] - \profile[2]}.
    \end{array}
\]
Via the accuracy assumption of $\cR$ and the sensitivity assumption of $\Delta_{\cR \circ \cA}$ (in Theorem~\ref{theorem: accuracy lower bound of profile estimation}), it holds that $m_B \doteq d \cdot \PAREN{ \estProfile[4] + \estProfile[0] - \estProfile[2] + 3 \cdot \Delta_{\cR \circ \cA} \cdot \LapNoise{1 / \eps} }$ is an estimate of $d \cdot \paren{\profile[4] + \profile[0] - \profile[2]}$ with expected error $o \paren{ \sqrt{d}  / {e^{c \cdot \eps}} }$:
\begin{equation*}
    \begin{array}{c}
        \E{\card{m_B - d \cdot \paren{\profile[4] + \profile[0] - \profile[2]}}}
            \in o \paren{ \sqrt{d}  / {e^{c \cdot \eps}} }.    
    \end{array}
\end{equation*}
Combing $d \cdot \paren{\profile[4] + \profile[0] - \profile[2]} = \inner{\vec{x}}{\vec{y}}$ with Markov inequality proves Equation~\eqref{eq: accuracy contradiction}.

Finally, we verify that $P$ is $\eps$-DP. 
Clearly $\cA(\vec{x} + \vecAllOne)$ is $\eps$-DP. 
The output of Bob, $d \cdot \paren{ \estProfile[4] + \estProfile[0] - \estProfile[2] + 3 \cdot \Delta_{\cR \circ \cA} \cdot \LapNoise{1 / \eps} }$ is also $\eps$-DP. 
This is because, by the definition of $\Delta_{\cR \circ \cA}$, it serves as an upper bound on the sensitivity of $\estProfile[4], \estProfile[0]$ and $\estProfile[2]$.

\vspace{-0.3cm}
\end{proof}

\vspace{-4mm}
We provide some further comments on the lower bound.
First, it is also possible to prove a lower bound of 
\begin{equation}
    \label{eq: alternative accuray lower bound}
    \begin{array}{c}
        \E{ \norm{ \profile - \cR \paren{ \cA(\hist) } }_\infty } 
        \in o \PAREN{ \frac{1}{\sqrt{d} \cdot \log d} \cdot \frac{1}{e^\eps} },
    \end{array}
\end{equation}
for Theorem~\ref{theorem: accuracy lower bound of profile estimation}, via a similar reduction to the two-party inner product lower bound accuracy lower bound by~\citet{McGregorMPRTV10}.
In particular, they showed that with probability \(1 - \delta\), inner product estimation in a two-party setting incurs an additive error of 
\(
    \Omega \PAREN{ \frac{\sqrt{d}}{\log d} \cdot \frac{\delta}{e^\eps} }
\).
Compared to Equation~\eqref{eq: accuray lower bound}, Equation~\eqref{eq: alternative accuray lower bound} avoids a constant $c$ blow up in the exponent, at the cost of an addition factor of $\log d$ in the denominator.

Second, Theorem~\ref{theorem: accuracy lower bound of profile estimation} assumes \(\Delta_{\cR \circ \cA} \in o \big( \frac{1}{\sqrt{d} } \cdot \frac{\eps}{e^{c \cdot \eps}} \big)\). We considered removing this assumption because it appears that if \(\Delta_{\cR \circ \cA} \in \Omega \big( \frac{1}{\sqrt{d} } \cdot \frac{\eps}{e^{c \cdot \eps}} \big)\), it would imply that 
\(
    \mathbb{E}[ \Vert \profile - \cR \paren{ \cA(\hist) } \Vert_\infty ]
            \in \Omega \big( \frac{1}{\sqrt{d} \cdot e^{c \cdot \eps}} \big)
\)
for some \(\hist\). However, this is not the case. Since \(\cA\) is a randomized algorithm, \(\Delta_{\cR \circ \cA}\) is defined over all randomness \(c_A\) of \(\cA\) as per Definition~\ref{def: composed sensitivity}. It is possible that \(\Delta_{\cR \circ \cA} \in \Omega \big( \frac{1}{\sqrt{d} } \cdot \frac{\eps}{e^{c \cdot \eps}} \big)\) for the \(c_A\)'s with very small probabilities. Therefore, we cannot derive a lower bound for 
\(
    \mathbb{E}[ \Vert \profile - \cR \paren{ \cA(\hist) } \Vert_\infty ].
\)

Finally, our current lower bound applies to an algorithm \(\cA\) that protects the number of occurrences of a single element by 1. 
It extends naturally to an algorithm \(\cA\) that protects the number of occurrences of a single element by an arbitrary amount, as long as it satisfies the other assumptions in Theorem~\ref{theorem: accuracy lower bound of profile estimation}.

\vspace{-4mm}
\section*{Impact Statement}%
\vspace{-2mm}
The goal of this work is to advance privacy-preserving data analysis and machine learning in distributed and large-scale settings.
Generally speaking this work adds to the toolbox of ways in which we can incorporate privacy into system designs.
There are many potential societal consequences of our work, none of which we feel must be specifically highlighted here.

\vspace{-4mm}
\section*{Acknowledgements}
\vspace{-2mm}
We thank the anonymous reviewers for their feedback which helped improve the paper.
The authors are supported by Providentia, a Data Science Distinguished Investigator grant from Novo Nordisk Fonden and affiliated with Basic Algorithms Research Copenhagen (BARC), supported by the VILLUM Foundation grant 16582.

\bibliography{reference}

\begin{thebibliography}{23}
\providecommand{\natexlab}[1]{#1}
\providecommand{\url}[1]{\texttt{#1}}
\expandafter\ifx\csname urlstyle\endcsname\relax
  \providecommand{\doi}[1]{doi: #1}\else
  \providecommand{\doi}{doi: \begingroup \urlstyle{rm}\Url}\fi

\bibitem[Acharya et~al.(2017)Acharya, Das, Orlitsky, and Suresh]{icmlAcharyaDOS17}
Acharya, J., Das, H., Orlitsky, A., and Suresh, A.~T.
\newblock A unified maximum likelihood approach for estimating symmetric properties of discrete distributions.
\newblock In \emph{International Conference on Machine Learning ({ICML})}, volume~70 of \emph{Proceedings of Machine Learning Research}, pp.\  11--21. {PMLR}, 2017.

\bibitem[Audibert et~al.(2009)Audibert, Munos, and Szepesv{\'{a}}ri]{AMS09}
Audibert, J., Munos, R., and Szepesv{\'{a}}ri, C.
\newblock Exploration-exploitation tradeoff using variance estimates in multi-armed bandits.
\newblock \emph{Theor. Comput. Sci.}, 410\penalty0 (19):\penalty0 1876--1902, 2009.

\bibitem[Blocki et~al.(2016)Blocki, Datta, and Bonneau]{BlockiDB16}
Blocki, J., Datta, A., and Bonneau, J.
\newblock Differentially private password frequency lists.
\newblock In \emph{Network and Distributed System Security Symposium ({NDSS})}. The Internet Society, 2016.

\bibitem[Bun et~al.(2019)Bun, Nelson, and Stemmer]{BNS19}
Bun, M., Nelson, J., and Stemmer, U.
\newblock Heavy hitters and the structure of local privacy.
\newblock \emph{{ACM} Trans. Algorithms}, 15\penalty0 (4):\penalty0 51:1--51:40, 2019.

\bibitem[Chen et~al.(2024)Chen, Indyk, and Woodruff]{chen_et_al:LIPIcs.ITCS.2024.32}
Chen, J.~Y., Indyk, P., and Woodruff, D.~P.
\newblock {Space-Optimal Profile Estimation in Data Streams with Applications to Symmetric Functions}.
\newblock In Guruswami, V. (ed.), \emph{15th Innovations in Theoretical Computer Science Conference (ITCS 2024)}, volume 287 of \emph{Leibniz International Proceedings in Informatics (LIPIcs)}, pp.\  32:1--32:22, Dagstuhl, Germany, 2024. Schloss Dagstuhl -- Leibniz-Zentrum f{\"u}r Informatik.
\newblock ISBN 978-3-95977-309-6.

\bibitem[Chung \& Lu(2006)Chung and Lu]{FL06}
Chung, F. and Lu, L.
\newblock Concentration inequalities and martingale inequalities: a survey.
\newblock \emph{Internet Math.}, 3\penalty0 (1):\penalty0 79--127, 2006.

\bibitem[Cormen et~al.(2009)Cormen, Leiserson, Rivest, and Stein]{CLRS09}
Cormen, T.~H., Leiserson, C.~E., Rivest, R.~L., and Stein, C.
\newblock \emph{Introduction to Algorithms, 3rd Edition}.
\newblock {MIT} Press, 2009.
\newblock ISBN 978-0-262-03384-8.

\bibitem[Datar \& Muthukrishnan(2002)Datar and Muthukrishnan]{datar2002estimating}
Datar, M. and Muthukrishnan, S.
\newblock Estimating rarity and similarity over data stream windows.
\newblock In \emph{Proceedings of European Symposium on Algorithms (ESA)}, pp.\  323--335. Springer, 2002.

\bibitem[Dwork \& Roth(2014)Dwork and Roth]{DR14}
Dwork, C. and Roth, A.
\newblock The algorithmic foundations of differential privacy.
\newblock \emph{Found. Trends Theor. Comput. Sci.}, 9\penalty0 (3-4):\penalty0 211--407, 2014.

\bibitem[Dwork et~al.(2006)Dwork, McSherry, Nissim, and Smith]{DworkMNS06}
Dwork, C., McSherry, F., Nissim, K., and Smith, A.~D.
\newblock Calibrating noise to sensitivity in private data analysis.
\newblock In Halevi, S. and Rabin, T. (eds.), \emph{Theory of Cryptography, Third Theory of Cryptography Conference, {TCC} 2006, New York, NY, USA, March 4-7, 2006, Proceedings}, volume 3876 of \emph{Lecture Notes in Computer Science}, pp.\  265--284. Springer, 2006.

\bibitem[Dwork et~al.(2010)Dwork, Naor, Pitassi, Rothblum, and Yekhanin]{DworkNPRY10}
Dwork, C., Naor, M., Pitassi, T., Rothblum, G.~N., and Yekhanin, S.
\newblock Pan-private streaming algorithms.
\newblock In \emph{Proceedings of Innovations in Computer Science (ICS)}, pp.\  66--80, 2010.
\newblock Now known as ITCS.

\bibitem[Ghosh et~al.(2009)Ghosh, Roughgarden, and Sundararajan]{GhoshRS09}
Ghosh, A., Roughgarden, T., and Sundararajan, M.
\newblock Universally utility-maximizing privacy mechanisms.
\newblock In Mitzenmacher, M. (ed.), \emph{Proceedings of the 41st Annual {ACM} Symposium on Theory of Computing, {STOC} 2009, Bethesda, MD, USA, May 31 - June 2, 2009}, pp.\  351--360. {ACM}, 2009.

\bibitem[Gray(2005)]{Gray05}
Gray, R.~M.
\newblock Toeplitz and circulant matrices: {A} review.
\newblock \emph{Found. Trends Commun. Inf. Theory}, 2\penalty0 (3), 2005.

\bibitem[Haitner et~al.(2022)Haitner, Mazor, Silbak, and Tsfadia]{HaitnerMST22}
Haitner, I., Mazor, N., Silbak, J., and Tsfadia, E.
\newblock On the complexity of two-party differential privacy.
\newblock In Leonardi, S. and Gupta, A. (eds.), \emph{{STOC} '22: 54th Annual {ACM} {SIGACT} Symposium on Theory of Computing, Rome, Italy, June 20 - 24, 2022}, pp.\  1392--1405. {ACM}, 2022.
\newblock \doi{10.1145/3519935.3519982}.
\newblock URL \url{https://doi.org/10.1145/3519935.3519982}.

\bibitem[Hardt \& Talwar(2010)Hardt and Talwar]{HardtT10}
Hardt, M. and Talwar, K.
\newblock On the geometry of differential privacy.
\newblock In Schulman, L.~J. (ed.), \emph{Proceedings of Symposium on Theory of Computing ({STOC})}, pp.\  705--714. {ACM}, 2010.

\bibitem[Hay et~al.(2009)Hay, Li, Miklau, and Jensen]{HayLMJ09}
Hay, M., Li, C., Miklau, G., and Jensen, D.~D.
\newblock Accurate estimation of the degree distribution of private networks.
\newblock In \emph{International Conference on Data Mining (ICDM)}, pp.\  169--178. {IEEE} Computer Society, 2009.

\bibitem[Hay et~al.(2010)Hay, Rastogi, Miklau, and Suciu]{HayRMS10}
Hay, M., Rastogi, V., Miklau, G., and Suciu, D.
\newblock Boosting the accuracy of differentially private histograms through consistency.
\newblock \emph{Proc. {VLDB} Endow.}, 3\penalty0 (1):\penalty0 1021--1032, 2010.

\bibitem[Karwa \& Slavkovic(2012)Karwa and Slavkovic]{KarwaS12}
Karwa, V. and Slavkovic, A.~B.
\newblock Differentially private graphical degree sequences and synthetic graphs.
\newblock In \emph{Privacy in Statistical Databases (PSD)}, volume 7556 of \emph{Lecture Notes in Computer Science}, pp.\  273--285. Springer, 2012.

\bibitem[Manurangsi(2022)]{Manurangsi22}
Manurangsi, P.
\newblock Tight bounds for differentially private anonymized histograms.
\newblock In \emph{5th Symposium on Simplicity in Algorithms}, pp.\  203--213. {SIAM}, 2022.

\bibitem[McGregor et~al.(2010)McGregor, Mironov, Pitassi, Reingold, Talwar, and Vadhan]{McGregorMPRTV10}
McGregor, A., Mironov, I., Pitassi, T., Reingold, O., Talwar, K., and Vadhan, S.~P.
\newblock The limits of two-party differential privacy.
\newblock In \emph{51th Annual {IEEE} Symposium on Foundations of Computer Science, {FOCS} 2010, October 23-26, 2010, Las Vegas, Nevada, {USA}}, pp.\  81--90. {IEEE} Computer Society, 2010.

\bibitem[Mitzenmacher \& Upfal(2017)Mitzenmacher and Upfal]{MU17}
Mitzenmacher, M. and Upfal, E.
\newblock \emph{Probability and computing: Randomization and probabilistic techniques in algorithms and data analysis}.
\newblock Cambridge university press, 2017.

\bibitem[Suresh(2019)]{Suresh19}
Suresh, A.~T.
\newblock Differentially private anonymized histograms.
\newblock In \emph{Advances in Neural Information Processing Systems 32}, pp.\  7969--7979, 2019.

\bibitem[Tolstikhin(2017)]{TIO17}
Tolstikhin, I.~O.
\newblock Concentration inequalities for samples without replacement.
\newblock \emph{Theory of Probability \& Its Applications}, 61\penalty0 (3):\penalty0 462--481, 2017.

\end{thebibliography}
\bibliographystyle{icml2024}

\newpage
\appendix
\onecolumn

\section{Probability Inequalities}

\begin{definition}[Lipschitz Condition] \label{def:Lipschitz-Condition} 
Consider an arbitrary set~$\cZ$.
A function $\Delta: \cZ^d \rightarrow \R$ satisfies the Lipschitz condition with bound $c \in \R$ if, for every $i \in [d]$ and for every sequence of values $x_1, \ldots, x_d \in \cZ$ and $y_i \in \R$, 
$$
    | \Delta( x_1, x_2, \ldots, x_{i - 1}, x_i, x_{i + 1}, \ldots, x_d ) - \Delta( x_1, x_2, \ldots, x_{i - 1}, y_i, x_{i + 1}, \ldots, x_d ) | \le c.
$$
\end{definition}

\begin{fact}[McDiarmid’s Inequality \citep{MU17, TIO17}] \label{fact:McDiarmid-Inequality}
Let $\Delta: \cZ^d \rightarrow \R$ be a function that satisfies the Lipschitz condition with bound~$c \in \R$. 
Let~$X_1, \ldots, X_d$ be independent random variables from the set~$\cZ$
Then for all $\lambda \ge 0$,
$$
    \P{ \Delta(X_1, \ldots, X_d) - \E{ \Delta(X_1, \ldots, X_d) } \ge \lambda } 
        \le \exp \PAREN{ - \frac{ 2 \lambda^2 } { d c^2 } }\,.   
$$
\end{fact}

\begin{fact}[Bernstein’s Inequality \citep{FL06, AMS09}] \label{fact: bernstein} 
    Let $X_1, \ldots, X_d$ be independent real-valued random variables such that $|X_i| \le c$ with probability one. 
    Let $S_d = \sum_{i \in [d]} X_i$ and $\Var{ S_d }  = \sum_{ i \in [d] } \Var{X_i}$. Then for all $\eta \in (0, 1)$, 
    $$
        \left| S_d - \E{ S_d } \right| 
            \le \sqrt{ 2 \Var{ S_d } \ln \frac{2}{\eta} } + \frac{2c \ln \frac{2}{\eta}}{ 3}\,,
    $$
    with probability at least $1 - \eta$. 
\end{fact}

\section{Proofs for Section~\ref{sec: unfolding clipped histogram}}
\label{appendix: sec: unfolding clipped histogram}

\begin{proof}[Proof of Lemma~\ref{lemma: recover clipped histogram}]
    Our goal is to prove for all~$t \in \Z,$
    \begin{align}
        \P{ \privHist'[\ell] = t } = \P{ \hist[\ell] + Z = t}.
    \end{align}
    
    The equation holds trivially for each integer~$t \in \PAREN{0, n}$.
    By symmetry, we prove that the equation holds for each integer~$t \ge n.$
    In this case, 
    \begin{align*}
        \P{ \privHist'[\ell] = t }
            &= \P{ \ExpNoise{e^{-\eps}} = t - n } \cdot \P{ \DiscreteLapNoise{e^{-\eps}} \ge n } \\
            &= \PAREN{ 
                    \PAREN{1 - e^{-\eps}} \cdot e^{-\eps \cdot \PAREN{t - n}}
                } \cdot \PAREN{
                    \sum_{j \ge n - \hist[\ell]} \frac{1 - e^{-\eps}}{1 + e^{-\eps}} \cdot e^{-\eps \cdot j}
                } \\
            &= \PAREN{1 - e^{-\eps}} \cdot e^{-\eps \cdot \PAREN{t - n}}
                \cdot \frac{1 - e^{-\eps}}{1 + e^{-\eps}} 
                \cdot \frac{e^{-\eps \cdot \PAREN{n - \hist[\ell]}}}{1 - e^{-\eps}} \\
            &= \frac{1 - e^{-\eps}}{1 + e^{-\eps}} 
                \cdot e^{-\eps \cdot \PAREN{t - \hist[\ell]}} \\
            &= \P{ \hist[\ell] + Z = t}.
    \end{align*}
\end{proof}

\section{Proofs for Section~\ref{sec: Searching for Profiles}}
\label{appendix: sec: searching for profiles}

\begin{proof}[Proof of Lemma~\ref{lemma: expectation in matrix form}]
    Let $\TDiscreteLapNoise{e^{-\eps}}{B}$ denotes a random variable following truncated Laplace distribution: 
    $$
        \P{ \TDiscreteLapNoise{e^{-\eps}}{B} = t } = \frac{1}{\ProbNorm} \cdot e^{-\eps \cdot \card{t} }, 
        \, \forall t \in \IntSet{-B}{B},
    $$
    where
    \begin{equation} 
        \ProbNorm 
        =  1 + 2 \cdot \sum_{j = 1}^B e^{- \eps \cdot j} 
        =  1 + 2 \cdot \frac{ 
                    e^{ -\eps } - (e^{ -\eps })^{ B + 1 }
                }{ 
                    1 - e^{ -\eps } 
                }
        = \frac{ 
                    1 + e^{ -\eps } - 2 \cdot (e^{ -\eps })^{ B + 1 }
                }{ 
                    1 - e^{ -\eps } 
                }.
    \end{equation}
    Conditioned on~$\cE_{B}$, all random noises introduced by Algorithm~\ref{algo: private profile generator} follows truncated Laplace distribution. 
    Therefore, 
    \begin{align*}
        \cA \PAREN{\estProfile}[t] 
            &= \frac{1}{d} \cdot \sum_{k = 0}^{t + B} \sum_{j = 1}^{\estProfile[k] \cdot d} \indicator{k + \TDiscreteLapNoise{e^{-\eps}}{B} = t},\,
            \quad \forall t \in \IntSet{-B}{-1}. \\
        \cA \PAREN{\estProfile}[t] 
            &= \frac{1}{d} \cdot \sum_{k = t - B}^{t + B} \sum_{j = 1}^{\estProfile[k] \cdot d} \indicator{k + \TDiscreteLapNoise{e^{-\eps}}{B} = t},\,
            \quad \forall t \in \IntSet{0}{n}. \\
        \cA \PAREN{\estProfile}[t] 
            &= \frac{1}{d} \cdot \sum_{k = t - B}^{n} \sum_{j = 1}^{\estProfile[k] \cdot d} \indicator{k + \TDiscreteLapNoise{e^{-\eps}}{B} = t},\,
            \quad \forall t \in \IntSet{n + 1}{n + B},         
    \end{align*}
    where~$\indicator{\cdot}$ is the indicator for some event.
    Since by assumption, we have~$\estProfile[t] = 0$ for all~$t \in \IntSet{-B}{n + B} \setminus \IntSet{0}{n}$, the above equations can be written in a unified manner: 
    \begin{align}
        \cA \PAREN{\estProfile}[t] 
            = \frac{1}{d} \cdot \sum_{k = t - B}^{t + B} \sum_{j = 1}^{\estProfile[k] \cdot d} \indicator{k + \TDiscreteLapNoise{e^{-\eps}}{B} \equiv t},
    \end{align}
    where for the ease of notations, we omit the modular operations on the indexes and the events of the indicators.
    In particular, we assume~$\estProfile[k] \doteq \estProfile[ \modularIndex{k} ]$, and~$k + \TDiscreteLapNoise{e^{-\eps}}{B} \equiv t$ represents the event that~$\modularIndex{k + \TDiscreteLapNoise{e^{-\eps}}{B}} = t$, where 
    $$
        \modularIndex{x} \doteq \begin{cases}
            x,  & \text{ if } x \in \IntSet{-B}{n + B}, \\
            \PAREN{ \PAREN{x + B} \mod \PAREN{n + 2B + 1} } - B & \text{ otherwise }.
        \end{cases}
    $$
    Combining the definition of~$\TransMatrix$, we can verify that
    \begin{align*}
        \E{\cA \PAREN{\estProfile}[t] \mid \cE_B }
                &=  \frac{1}{d} \cdot \frac{1}{\ProbNorm} \cdot \sum_{k = -B}^{B} e^{-\eps \cdot \card{k}} \cdot \PAREN{ \estProfile[t + k] \cdot d } \\
                &= \sum_{k = -B}^{B} \frac{1}{\ProbNorm} \cdot e^{-\eps \cdot \card{k}} \cdot \estProfile[t + k] \\
                &= \TransMatrix[t, \cdot] \times \estProfile.
    \end{align*}    
\end{proof}

\section{Proofs for Section~\ref{sec: optimization algorithm}}
\label{appendix: sec: optimization algorithm}

\begin{proof}[Proof of Lemma~\ref{lemma: correctness of algo fast inversion}]
    The proof consists of two parts. 
    First, we show that the returned solution~$\estProfile$ satisfies the constraint~$\inner{\vecTargetSupport}{\estProfile} = 1$: 
    \begin{align*}
            \inner{\vecTargetSupport}{\estProfile} 
            &= 
            \inner{
                \vecTargetSupport
            }{
                \vec{u} -  
                    \frac{
                        \inner{\vecTargetSupport}{\vec{u}} - 1
                    }{ 
                        \inner{\vecTargetSupport}{\TransMatrix^{-1} \times \vec{a}}
                    }  
                    \cdot 
                    \TransMatrix^{-1} \times \vec{a}'
            } \\ 
            &= 
            \inner{\vecTargetSupport}{\vec{u}} - \frac{
                    \inner{\vecTargetSupport}{\vec{u}} - 1
                }{ 
                    \inner{\vecTargetSupport}{\TransMatrix^{-1} \times \vec{a}}
                }  
                \cdot \inner{\vecTargetSupport}{\TransMatrix^{-1} \times \vec{a}} \\
            &= 1. 
    \end{align*}

    Next, we need to prove that~$\estProfile$ is an optimal solution.
    Consider an arbitrary vector $\vec{v},$ s.t., $\inner{\vecTargetSupport}{\vec{v}} = 1.$
    We want to prove that $\norm{\TransMatrix \times \vec{v} - \boundedProfileOfNoisyHist }_p \ge \norm{\TransMatrix \times \estProfile - \boundedProfileOfNoisyHist }_p.$
    The claim is trivial if~$\norm{\TransMatrix \times \estProfile - \boundedProfileOfNoisyHist }_p = 0$.
    Otherwise, based on Algorithm~\ref{algo: fast inversion}, observe that
    \begin{align*}
        \TransMatrix \times \estProfile - \boundedProfileOfNoisyHist
            = \TransMatrix \times \PAREN{ \estProfile - \vec{u} } 
            = \TransMatrix \times \PAREN{
                \frac{
                    { \inner{\vecTargetSupport}{\vec{u}} - 1 }
                }{ 
                    \inner{\vecTargetSupport}{\TransMatrix^{-1} \times \vec{a}}
                }  
                \cdot 
                \TransMatrix^{-1} \times \vec{a}.
            } 
            =   \frac{
                    { \inner{\vecTargetSupport}{\vec{u}} - 1 }
                }{ 
                    \inner{\vecTargetSupport}{\TransMatrix^{-1} \times \vec{a}}
                }  
                \cdot 
                \vec{a}.
    \end{align*}
    Since~$\norm{\vec{a}}_p = 1$, it holds that 
    $
        \frac{
            \TransMatrix \times \PAREN{ \estProfile - \vec{u} }
        }{
            \norm{\TransMatrix \times \PAREN{ \estProfile - \vec{u} }}_p
        } = \vec{a}.
    $
    Therefore, 
    \begin{align*}
        1 - \inner{\vecTargetSupport}{\vec{u}} 
            &= \inner{\vecTargetSupport}{\vec{v} - \vec{u}} \\
            &= \norm{\TransMatrix \times \PAREN{ \vec{v} - \vec{u} }}_p 
            \cdot \inner{\vecTargetSupport}{ 
                \TransMatrix^{-1} \times \frac{
                    \TransMatrix \times \PAREN{ \vec{v} - \vec{u} }
                }{
                    \norm{\TransMatrix \times \PAREN{ \vec{v} - \vec{u} }}_p
                } 
            }, 
        \\
        1 - \inner{\vecTargetSupport}{\vec{u}} 
            &= \inner{\vecTargetSupport}{\estProfile - \vec{u}} \\
            &= \norm{\TransMatrix \times \PAREN{ \estProfile - \vec{u} }}_p 
            \cdot \inner{\vecTargetSupport}{
                \TransMatrix^{-1} \times \frac{
                    \TransMatrix \times \PAREN{ \estProfile - \vec{u} }
                }{
                    \norm{\TransMatrix \times \PAREN{ \estProfile - \vec{u} }}_p
                } 
            } \\ 
            &= \norm{\TransMatrix \times \PAREN{ \estProfile - \vec{u} }}_p 
            \cdot \inner{\vecTargetSupport}{
                \TransMatrix^{-1} \times \vec{a}
            }.
    \end{align*}

    By the definition of~$\vec{a},$ we have $
        \inner{\vecTargetSupport}{
            \TransMatrix^{-1} \times \vec{a}
        }
        \ge 
        \inner{\vecTargetSupport}{ 
                \TransMatrix^{-1} \times \frac{
                    \TransMatrix \times \PAREN{ \vec{v} - \vec{u} }
                }{
                    \norm{\TransMatrix \times \PAREN{ \vec{v} - \vec{u} }}_p
                } 
            },
    $
    which proves that 
    \begin{align*}
        \norm{\TransMatrix \times \PAREN{ \estProfile - \vec{u} }}_p 
        &= \norm{\TransMatrix \times \estProfile - \boundedProfileOfNoisyHist }_p \\
        &\le \norm{\TransMatrix \times \vec{v} - \boundedProfileOfNoisyHist }_p
        = \norm{\TransMatrix \times \PAREN{ \vec{v} - \vec{u} }}_p.
    \end{align*}
\end{proof}

\begin{proof}[Proof of Lemma~\ref{lemma : correctness of algo rounding}]
    When~$\estProfile$ is inputted to Algorithm~\ref{algo: rounding}, it satisfies~$\sum_{t \in \IntSet{0}{n}} \estProfile[t] = 1$.
    After line~\ref{line: round to zero}, it satisfies~$\estProfile[t] = 0$, for all~$t \in \IntSet{-B}{n + B} \setminus \IntSet{0}{n}.$
    After line~\ref{line: round from below}, it satisfies~$\estProfile[t] \in [0, 1]$, for all~$t \in \IntSet{0}{n}$.
    However, it may violate is that~$\sum_{t \in \IntSet{0}{n}} \estProfile[t]$ should sum up to $1$. 
    
    Since the cumulative change to the sum~$\sum_{t \in \IntSet{0}{n}} \estProfile[t]$ is~$\surplus + \deficit$, after line~\ref{line: round from below}, we have~$\sum_{t \in \IntSet{0}{n}} \estProfile[t] - 1 = \surplus + \deficit$.
    We claim that~$\surplus + \deficit \ge 0$. 
    
    Since~$\deficit \ge 0$, the claim holds trivially if~$\surplus = 0.$
    On the other hand, if~$\surplus < 0$, then there exists some~$t \in \IntSet{0}{n}$, such that~$\estProfile[t] = 1$ after Line~\ref{line: round from above}. 
    However, after Line~\ref{line: round from above} and~\ref{line: round from below}, we must have that~$\estProfile[t] \in [0, 1], \forall t \in \IntSet{0}{n}.$
    Therefore, $\sum_{t \in \IntSet{0}{n}} \estProfile[t] \ge 1,$ and~$\surplus + \deficit = \sum_{t \in \IntSet{0}{n}} \estProfile[t] -1 \ge 0$. 

    In case that~$\surplus + \deficit > 0$, the algorithm needs to decreases the~$\estProfile[t]$ to fix the constraint that~$\sum_{t \in \IntSet{0}{n}} \estProfile[t]$ should sum up to $1$.
    Define the function~$g(\tau) \doteq \sum_{t \in \IntSet{0}{n}} \min \set{ \tau, \estProfile[t] }$. 
    The algorithm searches for the solution~$\tau$ for which~$g(\tau) = \surplus + \deficit$.
    By intermediate value theorem, the solution always exists, as
    \begin{itemize}
        \item~$g$ is a continuous function with respect to~$\tau$;
        \item~$0 = g(0) < \surplus + \deficit$;
        \item~$1 +  \surplus + \deficit = \sum_{t \in \IntSet{0}{n}} \estProfile[t] = g(1) > \surplus + \deficit.$
    \end{itemize}

    Finally the algorithm decreases each~$\estProfile[t]$ by~$\min \set{ \tau, \estProfile[t] }$ to fix the constraint so that they sum up to~$1$.
\end{proof}

\section{Proofs for Section~\ref{sec: Time Analysis}}
\label{appendix: sec: time analysis}

\begin{proof}[Proof of Lemma~\ref{lemma: property of transmatrix}]
    The proof relies on Fact~\ref{fact: eigenvalues and vectors of circulant matrix}.
    It's worth emphasizing that $\TransMatrix$ is a matrix with row and column indices starting from $-B$ (up to $n + B)$, in contrast to the matrix described in Fact~\ref{fact: eigenvalues and vectors of circulant matrix}, where indices commence from $0$.
    With appropriate shifting, we observe that the eigenvalues of $\TransMatrix$ are given by
    \resizebox{\linewidth}{!}{
        \begin{minipage}{\linewidth}
        \begin{align*}
                \eigenvalue_t 
                    \doteq \inner{\TransMatrix[-B, :]}{ \sqrt{n + 2B + 1} \cdot \vec{v}_t}
                &= \frac{1}{\ProbNorm} \cdot \PAREN{
                    \sum_{j = -B}^0 e^{- \eps \cdot \card{j + B}} \cdot w_t^{j + B} 
                    + \sum_{j = 1}^{n} 0 \cdot w_t^{j + B} 
                    + \sum_{j = n + 1}^{n + B} e^{- \eps \cdot \card{j - \PAREN{n + B + 1}}} \cdot w_t^{j + B} 
                } \\
                &= \frac{w_t^B}{\ProbNorm} \cdot \PAREN{
                    1  
                    + \sum_{j = 1}^B e^{- \eps \cdot j} \cdot w_t^{j} 
                    + \sum_{j = 1}^B e^{- \eps \cdot j} \cdot w_t^{-j} \cdot w_t^{n + 2B + 1} 
                } \\
                &= \frac{w_t^B}{\ProbNorm} \cdot \PAREN{ 
                    1 
                    + \frac{ 
                        (e^{ -\eps } \cdot  w_t )^1 - (e^{ -\eps } \cdot w_t )^{ B + 1 }
                    }{ 1 - e^{ -\eps } \cdot  w_t }
                    + \frac{ 
                        (e^{ -\eps } \cdot  w_t^{-1} )^1 - (e^{ -\eps } \cdot  w_t^{-1} )^{ B + 1 }
                    }{ 
                        1-e^{ -\eps } \cdot w_t^{-1} 
                    }
                } .
        \end{align*}
        \end{minipage}
    }
    Observe that 
    \begin{align*}
        &1 + \frac{ 
            (e^{ -\eps } \cdot  w_t )^1 - (e^{ -\eps } \cdot w_t )^{ \PAREN{ B + 1 } }
        }{ 1 - e^{ -\eps } \cdot  w_t }
        + 
        \frac{ 
            (e^{ -\eps } \cdot  w_t^{-1} )^1-(e^{ -\eps } \cdot  w_t^{-1} )^{ \PAREN{ B + 1 } }
        }{ 
            1-e^{ -\eps } \cdot w_t^{-1} 
        } \\
        &= 1 + \frac{ 
            e^{ -\eps } \cdot  w_t - e^{ -\eps \cdot \PAREN{ B + 1 } } \cdot w_t^{B} \cdot w_t^{1}
        }{ 1 - e^{ -\eps } \cdot  w_t }
        + 
        \frac{ 
            e^{ -\eps } \cdot  w_t^{-1} - e^{ -\eps \cdot \PAREN{ B + 1 } } \cdot w_t^{-B} \cdot w_t^{-1}
        }{ 
            1-e^{ -\eps } \cdot w_t^{-1} 
        } \\
        &= 1 + \frac{
            e^{ -\eps } \cdot w_t 
            - e^{-2 \eps }
            - e^{ -\eps \cdot \PAREN{ B + 1 } } \cdot w_t^{B} \cdot \PAREN{ w_t^{1} - e^{ -\eps }}
            + e^{ -\eps } \cdot w_t^{-1}
            - e^{-2 \eps }
            - e^{ -\eps \cdot \PAREN{ B + 1 } } \cdot w_t^{-B} \cdot \PAREN{ w_t^{-1} - e^{ -\eps }}
        }{
            1 - e^{ -\eps } \cdot w_t - e^{ -\eps } \cdot w_t^{-1} + e^{ -2 \eps } 
        } \\
        &= \frac{
                1 
                - e^{-2 \eps } 
                - e^{ -\eps \cdot \PAREN{ B + 1 } } \cdot \PAREN{ 
                    w_t^{B + 1} 
                    + w_t^{-(B + 1)} 
                    - e^{ -\eps } \cdot w_t^{B} 
                    - e^{ -\eps } \cdot w_t^{-B} 
                }
            }{
                1 - e^{ -\eps } \cdot w_t - e^{ -\eps } \cdot w_t^{-1} + e^{ -2 \eps } 
            }.
    \end{align*}
    It is easy to see that each~$\eigenvalue_t$ can be computed in~$O(1)$ time, therefore all the eigenvalues can be computed in~$O(n)$ time.
    This proves the first part of the lemma. 
    It follows that 
    \begin{align*}
        \TransMatrix \times \DFTMatrix
            &= \DFTMatrix \diag{\eigenvalue_{-B}, \eigenvalue_{-B + 1}, \ldots, \eigenvalue_{n + B} },  
        \\
        \TransMatrix 
            &= \DFTMatrix \times \diag{\eigenvalue_{-B}, \eigenvalue_{-B + 1}, \ldots, \eigenvalue_{n + B} }
            \times \PAREN{ \DFTMatrix }^{-1}.
    \end{align*}
    Since~$\eigenvalue_t \neq 0$ for all~$t \in \IntSet{-B}{n + B},$
    the inverse of~$\diag{\eigenvalue_{-B}, \eigenvalue_{-B + 1}, \ldots, \eigenvalue_{n + B} }$ is given by 
    $$
        \diag{\eigenvalue_{-B}, \eigenvalue_{-B + 1}, \ldots, \eigenvalue_{n + B} }^{-1}
            = \diag{\eigenvalue_{-B}^{-1}, \eigenvalue_{-B + 1}^{-1}, \ldots, \eigenvalue_{n + B}^{-1} }.
    $$
    Therefore, it can be verified that 
    \begin{align*}
        \TransMatrix^{-1}
            &= \DFTMatrix \times \diag{\eigenvalue_{-B}^{-1}, \eigenvalue_{-B + 1}^{-1}, \ldots, \eigenvalue_{n + B}^{-1} }
            \times \DFTMatrix^{-1}.
    \end{align*}

\end{proof}

\begin{proof}[Proof of Theorem~\ref{theorem: running time of algo rounding}]
    All steps in Algorithm~\ref{algo: rounding} clearly run in linear time, with the exception of the search for the $\tau$ such that~$\sum_{t \in \IntSet{0}{n}} \min \set{ \tau, \estProfile[t] } = s$.

    We show that it can be found in~$O(n \log n)$ time. 
    First, sort the~$\estProfile[t]$ in non-decreasing order in~$O(n \log n)$ time.
    Now, assume that~$\estProfile[0] \le \cdots \le \estProfile[n]$.
    After this step, our goal is to find the smallest~$t^* \in \IntSet{0}{n}$ such that~$\estProfile[t^*] \ge \tau$.
    Since
    \begin{equation*} 
    \begin{aligned}
        \sum_{t = 0}^{n} \min \set{ \tau, \estProfile[t] }  
            &= \sum_{t = 0}^{t^* - 1} \estProfile[t] + \sum_{t = t^*}^{n} \tau \\
            &= \sum_{t = 0}^{t^* - 1} \estProfile[t] + \PAREN{n + 1 - t^*} \tau = s,
    \end{aligned}
    \end{equation*}
    $t^*$ is the smallest number in~$\IntSet{0}{n}$ satisfying~$\sum_{t = 0}^{t^* - 1} \estProfile[t] + \PAREN{n + 1 - t^*} \estProfile[t^*] \ge s$ and therefore can be identified in~$O(n)$ time.
    Once~$t^*$ is known,~$\tau$ is determined by~$\tau = \PAREN{s - \sum_{t = 0}^{t^* - 1} \estProfile[t]} / \PAREN{n + 1 - t^*}.$

\end{proof}

\section{Proofs for Section~\ref{sec: error analysis}}
\label{appendix: sec: error analysis}

\begin{proof}[Proof of Lemma~\ref{lemma: error of rounding}]
    Since~$\profile$ is a vector in~$\R^{n + 1}$, we expand it into one in $\R^{n + 2B + 1}$ by padding $\profile[-B] = \profile[-B + 1] = \cdots = \profile[-1] = 0$ to the front and $\profile[n + 1] = \profile[n + 2] = \cdots = \profile[n + B] = 0$ to the back of the vector.
    After this,~$\profile$ has the same dimension as~$\estProfile$, therefore~$\norm{\estProfile - \profile}_p$ is well-defined. 
    
    We study how the error~$\norm{\estProfile - \profile}_p$ changes in Algorithm~\ref{algo: rounding}, for~$p = 1, 2, \infty.$
    To begin with, we divide Algorithm~\ref{algo: rounding} into 3 phrases: 
    \begin{itemize}
        \item Phrase 1 (Line~\ref{line: round to zero}): in which the algorithm rounds~$\estProfile[t]$ to~$0$ for each~$t \in \IntSet{-B}{n + B} \setminus \IntSet{0}{n}$.

        \item Phase 2 (Line~\ref{line: large set} to Line~\ref{line: round from below}), in which the algorithm clips the~$\estProfile[t]$ to the range of~$[0, 1]$. 

        \item Phase 3 (Line~\ref{line: compute adjust amount} to Line~\ref{line: end of fix}), in which the algorithm adjusts the~$\estProfile[t]$ so that they sum up to~$1$. 
    \end{itemize}

    The error only decreases in Phase 2 and 3: 
    \begin{itemize}
        \item 
        Since~$\profile[t] = 0$ for each~$t \in \IntSet{-B}{n + B} \setminus \IntSet{0}{n},$ rounding~$\estProfile[t]$ to~$0$ for each~$t \in \IntSet{-B}{n + B} \setminus \IntSet{0}{n}$ only decrease~$\norm{\estProfile - \profile}_p$.
    
        \item 
        Since~$\profile[t] \in [0, 1]$, for each~$t \in \IntSet{0}{n},$ clipping the~$\estProfile[t], t \in \IntSet{0}{n}$ to the range~$[0, 1]$ also decreases~$\norm{\estProfile - \profile}_p$.
        
    \end{itemize}

    It remains to study the change of~$\norm{\estProfile - \profile}_p$ in Phase 3.
    We discuss this for~$p = 1, 2, \infty$ separately. 
    To distinguish the value of $\estProfile$ at different phases of the algorithm, we denote it as $\estProfile_0$ before Phase 1, 
    $\estProfile_1$ right after Phase 1, 
    $\estProfile_2$ right after Phase 2, 
    and $\estProfile_3$ right after Phase 3.

    \vspace{3mm}
    {\it $\norm{\estProfile_3 - \profile}_1 \le \norm{\estProfile_1 - \profile}_1:$}
    observe that in Phase 2~$\norm{\estProfile - \profile}_1$ decreases by an amount of~$-\surplus + \deficit:$
    \begin{align*}
        \sum_{t \in \IntSet{0}{n}} \card{ \estProfile_1[t] - \profile[t]}
        &= \sum_{t \in \largeSet} \card{ \estProfile_1[t] - \profile[t]} 
        + \sum_{t \in \smallSet} \card{ \estProfile_1[1] - \profile[t]}
        + \sum_{t \in \IntSet{0}{n} \setminus \PAREN{\largeSet \cup \smallSet}} \card{ \estProfile_2[t] - \profile[t]} \\
        &= \sum_{t \in \largeSet} \PAREN{ \card{ 1 - \profile[t]} + {\estProfile_1[t] - 1} }
        + \sum_{t \in \smallSet} \PAREN{ \card{ 0 - \profile[t]} + {0 - \estProfile_1[t]} }
        + \sum_{t \in \IntSet{0}{n} \setminus \PAREN{\largeSet \cup \smallSet}} \card{ \estProfile_2[t] - \profile[t]} \\
        &= \sum_{t \in \IntSet{0}{n}} \card{ \estProfile_2[t] - \profile[t]} - \surplus + \deficit.
    \end{align*}
    In Phase 3,~$\norm{\estProfile - \profile}_1$ can increase by at most~$\surplus + \deficit$.
    To see that, observe that~$
        \sum_{t \in \IntSet{0}{n}} \estProfile_2[t] - 1
        = \sum_{t \in \IntSet{0}{n}} \estProfile_2[t] - \sum_{t \in \IntSet{0}{n}} \estProfile_1[t]
        = \surplus + \deficit
        = s. 
    $
    After the adjustment, we have~$\estProfile_3[t] \le \estProfile_2[t]$ for each $t \in \IntSet{0}{n}$, and
    $\sum_{t \in \IntSet{0}{n}} \estProfile_3[t] = 1$.
    Therefore, 
    \begin{align*}
        \sum_{t \in \IntSet{0}{n}} \card{ \estProfile_3[t] - \profile[t]}
        - \sum_{t \in \IntSet{0}{n}} \card{ \estProfile_2[t] - \profile[t]}  
        &\le \sum_{t \in \IntSet{0}{n}} \card{ \estProfile_2[t] - \estProfile_3[t]} \\
        &= \sum_{t \in \IntSet{0}{n}} \PAREN{ \estProfile_2[t] - \estProfile_3[t]} 
        =\sum_{t \in \IntSet{0}{n}} \estProfile_2[t] - 1 
        =\surplus + \deficit.
    \end{align*}
    It follows that
    \begin{align*}
        \sum_{t \in \IntSet{0}{n}} \card{ \estProfile_3[t] - \profile[t]}
        - \sum_{t \in \IntSet{0}{n}} \card{ \estProfile_1[t] - \profile[t]} 
            \le \surplus + \deficit - \PAREN{-\surplus + \deficit}
            = 2 \cdot \surplus \le 0. 
    \end{align*}
    
    {\it $\norm{\estProfile_3 - \profile}_\infty \le 2 \cdot \norm{\estProfile_2 - \profile}_\infty:$}
    Let~$b_t \doteq \card{\estProfile_2[t] - \profile[t]}$ for each~$t \in \IntSet{0}{n}$, and~$b_{\max} \doteq \max_{t \in \IntSet{0}{n}} b_t.$
    It is easy to see that 
    \begin{align*}
        \sum_{t \in \IntSet{0}{n}} \PAREN{ \estProfile_2[t] - \min \set{ b_{\max}, \estProfile_2[t] } }
            \le \sum_{t \in \IntSet{0}{n}} \profile[t] = 1. 
    \end{align*}
    On the other hand, we have 
    \begin{align*}
        \sum_{t \in \IntSet{0}{n}} \PAREN{ \estProfile_2[t] - \min \set{ \tau, \estProfile_2[t] } }  = 1. 
    \end{align*}
    It follows that~$\sum_{t \in \IntSet{0}{n}} \min \set{ \tau, \estProfile_2[t] } \le \sum_{t \in \IntSet{0}{n}} \min \set{ b_{\max}, \estProfile_2[t] },$ and hence~$\tau \le b_{\max}.$
    \begin{align*}
        \norm{\estProfile_3 - \profile}_\infty
            &= \max_{t \in \IntSet{0}{n}} \card{ \estProfile_2[t] - \min \set{ \tau, \estProfile_2[t] } - \profile[t] } \\
            &\le \max_{t \in \IntSet{0}{n}} \PAREN{ \card{ \estProfile_2[t] - \profile[t] } + \min \set{ \tau, \estProfile_2[t] } } \\
            &\le 2 \cdot \norm{\estProfile_2 - \profile}_\infty.
    \end{align*}
    
    {\it $\norm{\estProfile_3 - \profile}_2 \le \norm{\estProfile_2 - \profile}_2:$}
    without lose of generality, assume that~$\estProfile_2[0] \le \estProfile_2[1] \le \ldots \le \estProfile_2[n].$
    The adjustment step (line~\ref{line: adjust based on computed threshold}) can be replaced by Algorithm~\ref{algo: itereated algorithm}.
    \begin{algorithm}[!ht]
        \caption{Iterated Adjustment}
        \label{algo: itereated algorithm}
        \begin{algorithmic}[1]
            \STATE $t \leftarrow 0$
            \WHILE{$t \le n$ and~$s > 0$}
                \STATE $\tau' \leftarrow \min \PAREN{ \estProfile[t] , \frac{s}{n - t + 1} }$
                \STATE $\estProfile[j] \leftarrow \estProfile[j] - \tau',\, \forall j \in \IntSet{t}{n}$
                \STATE $s \leftarrow s - \tau' \cdot \PAREN{n - t + 1}$
                \STATE $t \leftarrow t + 1$
            \ENDWHILE
        \end{algorithmic}
    \end{algorithm}

    We claim that, 
    $\norm{\estProfile - \profile}_2$ does not increase at each iteration.
    Fix some iteration~$t$.
    Note that~$\estProfile[j]$ does not change for~$j \in \IntSet{0}{t - 1}$, whereas~$\estProfile[j]$ decreases by~$\tau'$ for~$j \in \IntSet{t}{n}$.
    Therefore, the total change of~$\norm{\estProfile - \profile}_2^2$ is given by 
    \begin{align*}
        &\sum_{j \in \IntSet{t}{n}} \card{\estProfile[t] - \tau' - \profile[t]}^2 
            - \sum_{j \in \IntSet{t}{n}} \card{\estProfile[t] - \profile[t]}^2 \\
        &= \sum_{j \in \IntSet{t}{n}} -\tau' \cdot \PAREN{2 \cdot \estProfile[t] - 2 \cdot \profile[t] - \tau' } \\
        &= -\tau' \cdot \PAREN{2 \cdot \sum_{j \in \IntSet{t}{n}} \estProfile[t] - 2 \cdot \sum_{j \in \IntSet{t}{n}} \profile[t] - \sum_{j \in \IntSet{t}{n}} \tau' } \\
        &= -\tau' \cdot \PAREN{2 \cdot \PAREN{ \sum_{j \in \IntSet{t}{n}} \estProfile[t] - 1 } - \PAREN{n - t + 1} \cdot \tau' } \\
        &\le 0,
    \end{align*}
    where the last inequality follows since~$0 \le \tau' \le \estProfile[j]$ for each~$j \in \IntSet{t}{n}.$
\end{proof}

We present the proof for Lemma~\ref{lemma: deviation of profile of noisy histogram} first, then the proof for Lemma~\ref{lemma: bounds on matrix norms}. 

\begin{proof}[Proof of Lemma~\ref{lemma: deviation of profile of noisy histogram}]
    Since~$\profile$ is a vector in~$\R^{n + 1}$, we expand it into one in $\R^{n + 2B + 1}$ by padding $\profile[-B] = \profile[-B + 1] = \cdots = \profile[-1] = 0$ to the front and $\profile[n + 1] = \profile[n + 2] = \cdots = \profile[n + B] = 0$ to the back of the vector, so that~$\TransMatrix \times \profile$ is well-defined. 

    Recall that
    $
        \privHist[\ell] = \hist[\ell] + \DiscreteLapNoise{e^{-\eps}},\, 
        \forall \ell \in \DataDomain.
    $
    Since $B \ge \frac{1}{\eps} \cdot \ln \frac{2 \cdot d}{\eta \cdot \PAREN{e^\eps + 1}}$.
    Then  
    \begin{align*}
        \P{\card{\DiscreteLapNoise{e^{-\eps}}} > B}
            &= 2 \cdot \sum_{t \ge B} \frac{1 - e^{-\eps}}{1 + e^{-\eps}} \cdot e^{-\eps \cdot \card{t} }
            = 2 \cdot \frac{1 - e^{-\eps}}{1 + e^{-\eps}} \cdot \frac{e^{-\eps \cdot B}}{1 - e^{-\eps}}
            = \frac{2 \cdot e^{-\eps \cdot B}}{1 + e^{-\eps}}
            \le \frac{\eta}{d}.
    \end{align*}
    By union bound,
    \begin{align*}
        \P{ \max_{\ell \in \DataDomain} \card{\privHist[\ell] - \hist[\ell]} > B } \le \eta.   
    \end{align*}

    Assuming that~$\max_{\ell \in \DataDomain} \card{\privHist[\ell] - \hist[\ell]} \le B$, $\boundedProfileOfNoisyHist = \cA \paren{\profile}$ can be equivalent viewed as the outcome of running Algorithm~\ref{algo: private profile generator} with input~$\profile$, (conditioned on the event~$\cE_{B}$ that all noises added by the algorithm are absolutely bounded by~$B$). 
    It follows by Lemma~\ref{lemma: expectation in matrix form} that 
    \begin{align*}
        \E{\boundedProfileOfNoisyHist} =  
            \E{\cA \PAREN{\profile} \mid \cE_B } 
            = \TransMatrix \times \profile.
    \end{align*}
    
    {\bf Bounding~$\norm{ \boundedProfileOfNoisyHist - \TransMatrix \times \profile }_1$.}
    Under these assumptions, 
    \begin{align*}
        \E{ \norm{ \boundedProfileOfNoisyHist - \TransMatrix \times \profile }_1 }
            &= \E{ \norm{ \boundedProfileOfNoisyHist - \E{ \boundedProfileOfNoisyHist } }_1 }
            = \E{ \sum_{t = -B}^{n + B} \card{ \boundedProfileOfNoisyHist[t] - \E{ \boundedProfileOfNoisyHist[t] } } } \\
            &= \sum_{t = -B}^{n + B} \E{ \card{ \boundedProfileOfNoisyHist[t] - \E{ \boundedProfileOfNoisyHist[t] } } }
            = \sum_{t = -B}^{n + B} \E{ \sqrt{ \card{ \boundedProfileOfNoisyHist[t] - \E{ \boundedProfileOfNoisyHist[t] } }^2 } } \\
            &\le \sum_{t = -B}^{n + B} \sqrt{ \E{ \card{ \boundedProfileOfNoisyHist[t] - \E{ \boundedProfileOfNoisyHist[t] } }^2 } }. 
    \end{align*}
    
    Consider a $t \in \IntSet{-B}{n + B}$. 
    Conditioned on the event that~$\max_{\ell \in \DataDomain} \card{\privHist[\ell] - \hist[\ell]} \le B$, the noises added to the~$\privHist[\ell]$ follows independent truncated Laplace distribution.
    For each $\ell \in \DataDomain$, define the indicator $\indicator{\privHist[\ell] = t}$ for the event that $\privHist[\ell] = t$. 
    We have
    $$
        \P{ \indicator{\privHist[\ell] = t} = 1 } 
            = \frac{ 
                e^{-\eps \cdot \card{ t - \hist[\ell] }}
            }{ 
                \ProbNorm
            },
    $$
    where~$\ProbNorm = \frac{ 
                    1 + e^{ -\eps } - 2 \cdot (e^{ -\eps })^{ B + 1 }
                }{ 
                    1 - e^{ -\eps } 
                }.$ 
    Since $
        \boundedProfileOfNoisyHist [t] = \frac{ 1 }{ d } \cdot \sum_{\ell \in \DataDomain} \indicator{\privHist[\ell] = t},
    $ 
    we have 
    \begin{align*}
        \E{ 
            \card{ 
                \boundedProfileOfNoisyHist[t] - \E{ \boundedProfileOfNoisyHist[t] } 
            }^2 
        }
            &= \Var{
                \frac{ 1 }{ d } \cdot \sum_{\ell \in \DataDomain} \indicator{\privHist[\ell] = t}
            } \\
            &= \frac{1}{d^2} \cdot \sum_{\ell \in \DataDomain} \Var{ \indicator{\privHist[\ell] = t} } \\
            &= \frac{1}{d^2} \cdot  
                \sum_{\ell \in \DataDomain} \frac{ 
                            e^{-\eps \cdot \card{t - \hist[\ell]} } 
                        }{\ProbNorm} 
                        \cdot \PAREN{ 
                            1 - \frac{
                                e^{-\eps \cdot \card{t - \hist[\ell]} }
                            }{\ProbNorm}
                        } \\
            &\le \frac{1}{ d } \cdot \E{ \boundedProfileOfNoisyHist [t] }. 
    \end{align*}
    Note that $
        1 - \frac{ e^{-\eps \cdot \card{t - \hist[\ell]} } }{\ProbNorm} 
        \ge 1 - \frac{1}{ \ProbNorm } 
        = 1 - \frac{ 1 - e^{-\eps} }{  1 + e^{ -\eps } - 2 \cdot (e^{ -\eps })^{ B + 1 } }
        = \frac{ 2 e^{-\eps} - 2 \cdot (e^{ -\eps })^{ B + 1 } }{  1 + e^{ -\eps } - 2 \cdot (e^{ -\eps })^{ B + 1 } }
    $.
    So the inequality is tight up to a factor of $
        \frac{ 2 e^{-\eps} - 2 \cdot (e^{ -\eps })^{ B + 1 } }{  1 + e^{ -\eps } - 2 \cdot (e^{ -\eps })^{ B + 1 } }
    $.
    It concludes that 
    \begin{align*}
        \E{ \norm{ \boundedProfileOfNoisyHist - \TransMatrix \times \profile }_1 }
            \le \sum_{t = -B}^{n + B} \sqrt{ \E{ \card{ \boundedProfileOfNoisyHist[t] - \E{ \boundedProfileOfNoisyHist[t] } }^2 } }
            = \sum_{t = -B}^{n + B} \sqrt{ \frac{1}{ d } \cdot \E{ \boundedProfileOfNoisyHist [t] } }.
    \end{align*}

    Further, $\norm{ \boundedProfileOfNoisyHist - \TransMatrix \times \profile }_1$ is a function of the random variables~$\privHist[\ell]'s:$
    \begin{align*}
        \norm{ \boundedProfileOfNoisyHist - \TransMatrix \times \profile }_1
            &= \sum_{t \in \IntSet{-B}{n + B}} \card{
                 \frac{ 1 }{ d } \cdot \sum_{\ell \in \DataDomain} \indicator{\privHist[\ell] = t} - \E{ \boundedProfileOfNoisyHist [t] }
            }.
    \end{align*}
    Changing the value of one~$\privHist[\ell]$ affects at most two terms on the right hand side of the equation, each by~$1 / d$.
    By Definition~\ref{def:Lipschitz-Condition}, $\norm{ \boundedProfileOfNoisyHist - \TransMatrix \times \profile }_1$ satisfies the Lipschitz condition with bound~$2 / d$.
    Therefore, via Fact~\ref{fact:McDiarmid-Inequality}, 
    \begin{align*}
        \P{
            \norm{ \boundedProfileOfNoisyHist - \TransMatrix \times \profile }_1
            - \E{\norm{ \boundedProfileOfNoisyHist - \TransMatrix \times \profile }_1}
            \ge \sqrt{ \frac{2 \ln \frac{1}{\eta} }{d} }
        }
        \le \exp \PAREN{
            -\frac{2}{d \cdot \PAREN{2 / d}^2} \cdot \frac{2 \ln \frac{1}{\eta} }{d}
        }
        = \eta. 
    \end{align*}

    {\bf Bounding~$\norm{ \boundedProfileOfNoisyHist - \TransMatrix \times \profile }_2$.}
    Similarly, 
     \begin{align*}
        \E{ \norm{ \boundedProfileOfNoisyHist - \TransMatrix \times \profile }_2 }
            &= \E{ \norm{ \boundedProfileOfNoisyHist - \E{ \boundedProfileOfNoisyHist } }_2 }
            = \E{ \sqrt{ \sum_{t = -B}^{n + B} \card{ \boundedProfileOfNoisyHist[t] - \E{ \boundedProfileOfNoisyHist[t] } }^2 } } \\
            &\le \sqrt{ \sum_{t = -B}^{n + B} \E{ \card{ \boundedProfileOfNoisyHist[t] - \E{ \boundedProfileOfNoisyHist[t] } }^2 } }
            \le \sqrt{ \sum_{t = -B}^{n + B} \frac{1}{ d } \cdot \E{ \boundedProfileOfNoisyHist [t] } } 
            = \sqrt{\frac{1}{d}}.
    \end{align*}

    Further, $\norm{ \boundedProfileOfNoisyHist - \TransMatrix \times \profile }_2$ is a function of the random variables~$\privHist[\ell]:$
    \begin{align*}
        \norm{ \boundedProfileOfNoisyHist - \TransMatrix \times \profile }_2
            &= \sqrt{ \sum_{t \in \IntSet{-B}{n + B}} \card{
                 \frac{ 1 }{ d } \cdot \sum_{\ell \in \DataDomain} \indicator{\privHist[\ell] = t} - \E{ \boundedProfileOfNoisyHist [t] }
            }^2 },
    \end{align*}
    We can view~$\boundedProfileOfNoisyHist = \PAREN{ \PAREN{1 / d} \cdot \sum_{\ell \in \DataDomain} \indicator{\privHist[\ell] = -B}, \ldots, \PAREN{1 / d} \cdot \sum_{\ell \in \DataDomain} \indicator{\privHist[\ell] = n + B} }$.
    Assume that we change the value of one~$\privHist[\ell]$, to obtain a vector~$\boundedProfileOfNoisyHist',$ 
    which differs from~$\boundedProfileOfNoisyHist$ by two coordinates, each by~$1 / d$.
    Therefore,~$\norm{ \boundedProfileOfNoisyHist' - \boundedProfileOfNoisyHist }_2 = \sqrt{2} / d$, and 
    \begin{align*}
        \card{ 
            \norm{ \boundedProfileOfNoisyHist' - \TransMatrix \times \profile }_2
            - \norm{ \boundedProfileOfNoisyHist - \TransMatrix \times \profile }_2
        } \le 
        \norm{
            \boundedProfileOfNoisyHist' - \boundedProfileOfNoisyHist
        }_2
        \le \sqrt{2} / d. 
    \end{align*}
    By Definition~\ref{def:Lipschitz-Condition}, $\norm{ \boundedProfileOfNoisyHist - \TransMatrix \times \profile }_2$ satisfies the Lipschitz condition with bound~$\sqrt{2} / d$.
    Therefore, via Fact~\ref{fact:McDiarmid-Inequality}, 
    \begin{align*}
        \P{
            \norm{ \boundedProfileOfNoisyHist - \TransMatrix \times \profile }_2
            - \E{\norm{ \boundedProfileOfNoisyHist - \TransMatrix \times \profile }_2}
            \ge \sqrt{ \frac{\ln \frac{1}{\eta} }{d} }
        }
        \le \exp \PAREN{
            -\frac{2}{d \cdot \PAREN{\sqrt{2} / d}^2} \cdot \frac{\ln \frac{1}{\eta} }{d}
        }
        = \eta. 
    \end{align*}

    {\bf Bounding~$\norm{ \boundedProfileOfNoisyHist - \TransMatrix \times \profile }_\infty$.}
    To prove the last part, 
    by Bernstein's inequality (Fact~\ref{fact: bernstein}), with probability at least $1 - 1 / (n \eta)$, we have 
    \begin{align*}
        \card{ \boundedProfileOfNoisyHist [t] - \E{ \boundedProfileOfNoisyHist [t] }} 
            &\le \sqrt{ 
                    2 \cdot 
                       \Var{
                            \boundedProfileOfNoisyHist [t]
                        }
                    \cdot \ln \frac{n}{\eta} 
                } 
                + \frac{2}{3} \cdot \ln \frac{n}{\eta} \\
            &\le \sqrt{ 
                \frac{2}{d} \cdot \PAREN{ 
                        \sum_{\ell \in \DataDomain} \frac{ 
                            e^{-\eps \cdot \card{t - \hist[\ell]} } 
                        }{\ProbNorm} 
                    } \cdot \ln \frac{n}{\eta} 
            } + \frac{2}{ 3 d } \cdot \ln \frac{n}{\eta} \\
            &= \frac{1}{ \sqrt{d} } \cdot \sqrt{ 2 \cdot \E{ \boundedProfileOfNoisyHist [t] } \cdot \ln \frac{n}{\eta} } + \frac{2}{ 3 d } \cdot \ln \frac{n}{\eta} .
    \end{align*}
    
    It is not easy to have a non-trivial bound for $\E{ \boundedProfileOfNoisyHist [t] }$, which depends on the histogram~$\hist$.
    On the other hand, a trivial bound can be given by 
    $$
         \E{ \boundedProfileOfNoisyHist [t] }
            \le \frac{1}{d \cdot \ProbNorm} \cdot \sum_{\ell \in \DataDomain} e^{-\eps \cdot 0 } 
            = \frac{1}{\ProbNorm}.
    $$

\end{proof}

\begin{proof}[Proof of Lemma~\ref{lemma: bounds on matrix norms}]
    We first bound~$\norm{ \TransMatrix^{-1} }_2$, since it is the easier part of the proof.
    According to Lemma~\ref{lemma: property of transmatrix},
    \begin{align}
        \norm{ \TransMatrix^{-1} }_2 
        &= \norm{ 
            \DFTMatrix \times 
            \diag{\eigenvalue_{-B}^{-1}, \eigenvalue_{-B + 1}^{-1}, \ldots, \eigenvalue_{n + B}^{-1} }
            \times \DFTMatrix^{-1}
        }_2 \\
        &\le \norm{ \DFTMatrix }_2 \cdot 
        \norm{ 
            \diag{\eigenvalue_{-B}^{-1}, \eigenvalue_{-B + 1}^{-1}, \ldots, \eigenvalue_{n + B}^{-1} } 
        }_2 
        \cdot \norm{ \DFTMatrix^{-1} }_2. 
    \end{align}
    Since~$\DFTMatrix$ is the discrete Fourier transform matrix with orthonormal column and row vectors, it holds that~$\norm{ \DFTMatrix }_2 = \norm{ \DFTMatrix^{-1} }_2 = 1$.
    Moreover, $\norm{ \diag{\eigenvalue_{-B}^{-1}, \eigenvalue_{-B + 1}^{-1}, \ldots, \eigenvalue_{n + B}^{-1} } }_2 = \max_{t \in \IntSet{-B}{n + B}} \frac{1}{|\eigenvalue_t|}$.
    We proceed to show that to bounds~$\max_{t \in \IntSet{-B}{n + B}} \frac{1}{|\eigenvalue_t|} = 1 / \min_{t \in \IntSet{-B}{n + B}} |\eigenvalue_t|$. 
    For each~$t \in \IntSet{-B}{n + B}$, 
    \begin{align*}
        \card{\eigenvalue_t}
        =
        \card{ 
            \frac{w_t^B}{\ProbNorm} \cdot \frac{
                1 
                - e^{-2 \eps } 
                - e^{ -\eps \cdot \PAREN{ B + 1 } } \cdot \PAREN{ 
                    w_t^{B + 1} 
                    + w_t^{-(B + 1)} 
                    - e^{ -\eps } \cdot w_t^{B} 
                    - e^{ -\eps } \cdot w_t^{-B} 
                }
            }{
                1 - e^{ -\eps } \cdot w_t - e^{ -\eps } \cdot w_t^{-1} + e^{ -2 \eps } 
            }
        } 
    \end{align*}
    We have 
    \begin{align*}
        \card{w_t^B} 
            &= 1. \\
        \card{  1 - e^{ -\eps } \cdot w_t - e^{ -\eps } \cdot w_t^{-1} + e^{ -2 \eps }  }
            &\le \card{ 1 + 2 \cdot e^{ -\eps } + e^{ -2 \eps } }
            = \PAREN{ 1 + e^{ -\eps } }^2. 
    \end{align*}
    and
    \begin{align*}
        &\card{ 
            1 
            - e^{-2 \eps } 
            - e^{ -\eps \cdot \PAREN{ B + 1 } } \cdot \PAREN{ 
                w_t^{B + 1} 
                + w_t^{-(B + 1)} 
                - e^{ -\eps } \cdot w_t^{B} 
                - e^{ -\eps } \cdot w_t^{-B} 
            }
        }   \\
        &\ge \card{ 
            1 
            - e^{-2 \eps } 
        } + e^{ -\eps \cdot \PAREN{ B + 1 } } \cdot \card{ 
                w_t^{B + 1} 
                + w_t^{-(B + 1)} 
                - e^{ -\eps } \cdot w_t^{B} 
                - e^{ -\eps } \cdot w_t^{-B} 
            }\\
        &\ge 
            1 
            - e^{-2 \eps } 
            - e^{ -\eps \cdot \PAREN{ B + 1 } } \cdot 2 \cdot \PAREN{ 1
                + e^{ -\eps } 
            } \\
            &= \PAREN{
                1 - e^{ -\eps } - 2 \cdot e^{ -\eps \cdot \PAREN{ B + 1 } } 
            } \cdot \PAREN{ 1 + e^{ -\eps } }.
    \end{align*}
    Therefore
    $$
        \card{\eigenvalue_t} 
            \ge \frac{ 
                    1
                }{
                    \ProbNorm
                } \cdot  
            \frac{
                \PAREN{
                    1 - e^{ -\eps } - 2 \cdot e^{ -\eps \cdot \PAREN{ B + 1 } } 
                } 
            }{
                \PAREN{ 1 + e^{ -\eps } }
            }.
    $$
    \\

    To bound~$\norm{ \TransMatrix^{-1} }_1$ and~$\norm{ \TransMatrix^{-1} }_\infty$, we will first show that there two values are equal; then we bound just~$\norm{ \TransMatrix^{-1} }_\infty$. \\

    {\it $\norm{ \TransMatrix^{-1} }_1 = \norm{ \TransMatrix^{-1} }_\infty.$}
    We will show that,~$\TransMatrix^{-1}$ is also a circulant matrix. 
    Since each column of a circulant matrix has the same $\ell_1$-norm as each row of the matrix, and since 
    \begin{align*}
        \norm{ \TransMatrix^{-1} }_1 
            &= \max_{t \in \IntSet{-B}{n + B}} \norm{ \TransMatrix^{-1}[\cdot , t] }_1, \\
        \norm{ \TransMatrix^{-1} }_\infty 
            &= \max_{t \in \IntSet{-B}{n + B}} \norm{ \TransMatrix^{-1}[t, \cdot ] }_1,
    \end{align*}
    it follows that 
    $
        \norm{ \TransMatrix^{-1} }_1 
            = \norm{ \TransMatrix^{-1} }_\infty. 
    $
    It remains to show that~$\TransMatrix^{-1}$ is also a circulant matrix. 
    According to Lemma~\ref{lemma: property of transmatrix},~$\DFTMatrix$ is the discrete Fourier transform  matrix whose rows and columns are indexed from~$-B$ to~$n + B$.
    Therefore~$\forall k, \ell \in \IntSet{-B}{n + B}$, it holds that 
    \begin{align*}
        \DFTMatrix[k , \ell] 
            &= \frac{1}{\sqrt{n + 2B + 1}} \cdot \eigenvalue_\ell^{k + B} \\
            &= \frac{1}{\sqrt{n + 2B + 1}} \cdot e^{- \frac{2 \pi \PAREN{k + B} \cdot \PAREN{\ell + B} \cdot i}{n + 2B + 1} }, 
        \, \\
        \DFTMatrix^{-1}[k , \ell] 
            &= \frac{1}{\sqrt{n + 2B + 1}} \cdot \eigenvalue_k^{-\PAREN{\ell + B}} \\
            &= \frac{1}{\sqrt{n + 2B + 1}} \cdot e^{ \frac{2 \pi \PAREN{k + B} \cdot \PAREN{\ell + B} \cdot i}{n + 2B + 1} }.
    \end{align*}
    Since 
    $
        \TransMatrix^{-1}
            = \DFTMatrix 
            \times \diag{\eigenvalue_{-B}^{-1}, \eigenvalue_{-B + 1}^{-1}, \ldots, \eigenvalue_{n + B}^{-1} }
            \times \DFTMatrix^{-1},
    $
    we can verify that  
    \begin{align*}
        \TransMatrix^{-1}[k, \ell] 
            &= \DFTMatrix[k, \cdot] \times \diag{\eigenvalue_{-B}^{-1}, \eigenvalue_{-B + 1}^{-1}, \ldots, \eigenvalue_{n + B}^{-1} } \times \DFTMatrix^{-1}[\cdot, \ell] \\
            &= \frac{1}{n + 2B + 1} \cdot \inner{
                \PAREN{
                    e^{- \frac{2 \pi \PAREN{k + B} \cdot 0 \cdot i}{n + 2B + 1} } \eigenvalue_{-B}^{-1},
                    \ldots, 
                    e^{- \frac{2 \pi \PAREN{k + B} \cdot \PAREN{n + 2B} \cdot i}{n + 2B + 1} } \eigenvalue_{n + B}^{-1}
                }
            }{
                \PAREN{
                    e^{\frac{2 \pi {0 \cdot \PAREN{\ell + B}} \cdot i}{n + 2B + 1} },
                    \ldots, 
                    e^{\frac{2 \pi {\PAREN{n + 2B} \cdot \PAREN{\ell + B}} \cdot i}{n + 2B + 1} }
                }
            } \\
            &= \frac{1}{n + 2B + 1} \cdot \sum_{t \in \IntSet{-B}{n + B}} e^{- \frac{2 \pi \PAREN{k - \ell}  \cdot \PAREN{t + B} \cdot i}{n + 2B + 1} } \eigenvalue_t^{-1}.
    \end{align*}
    Now it is easy to see that~$\TransMatrix^{-1}$ is a circulant matrix, since for each~$-B \le k < n + B$, we have 
    \begin{align*}
        \TransMatrix^{-1}[k + 1, \ell + 1] 
            &= \TransMatrix^{-1}[k, \ell], \quad \forall \ell \in \IntSet{-B}{n + B - 1}, 
        \\
        \TransMatrix^{-1}[k + 1, -B] 
            &= \TransMatrix^{-1}[k, n + B].
    \end{align*}

    {\it Bounding $\norm{ \TransMatrix^{-1} }_\infty.$}
    We will show that for each~$\vec{v} \in \R^{n + 2B + 1}$, it holds that 
    $$
        \norm{ \TransMatrix^{-1} \vec{v} }_1 \le  \frac{
                2 + e^{ -\eps } + e^{ \eps }
            }{
                e^{ \eps } - e^{ -\eps } - 4 \cdot e^{-\eps \cdot B }
            } 
            \cdot \ProbNorm 
            \cdot \norm{\vec{v}}_1.
    $$
    Since~$\TransMatrix$ is invertible, let $\vec{x} \doteq \TransMatrix^{-1} \vec{v}$. 
    It reduces to prove that 
    $$
        \norm{ \vec{x} }_1 \le  \frac{
                2 + e^{ -\eps } + e^{ \eps }
            }{
                e^{ \eps } - e^{ -\eps } - 4 \cdot e^{-\eps \cdot B }
            } 
            \cdot \ProbNorm 
            \cdot \norm{\TransMatrix \vec{x}}_1.
    $$
    To finish the proof, we need the following lemma. \\

    \begin{lemma}
    \label{lemma: bounds on inverting a single entry}
        Let~$\vec{x} \in \R^{n + 2B + 1}$ and assume that for each~$t \in \IntSet{-B}{n + B}$, there exists~$b_t \in \R_{\ge 0}$, s.t.~$\card{ {\TransMatrix[t, \cdot] \times \vec{x}}} \le b_t$.
        The for each~$t \in \IntSet{-B}{n + B}$, we have 
        
        \vspace{-3mm}
        \resizebox{\linewidth}{!}{
          \begin{minipage}{1.13\linewidth}
            \begin{equation*}
                \card{\vec{x}[t]} \le \frac{
                    e^{ -\eps } + e^{ \eps }
                }{
                    e^{ \eps } - e^{ -\eps }
                } \cdot \PAREN{ 
                    \ProbNorm \cdot \PAREN{
                        b_t + \frac{b_{t - 1} + b_{t + 1}}{e^{ -\eps } + e^{ \eps}}
                    }  
                    +
                     \frac{e^{-\eps \cdot B }}{e^{ -\eps } + e^{ \eps } } \cdot \PAREN{ 
                        \card{ \vec{x} \bracket{t - B - 1} } + \card{ \vec{x} \bracket{t + B + 1} } + e^{ -\eps } \cdot \card{ \vec{x} \bracket{t - B} } + e^{ -\eps } \cdot \card{ \vec{x} \bracket{t + B} }
                    }
                },
            \end{equation*}
          \end{minipage}
        }
        
        where for simplicity we omit the modular operations on the subscripts and indexes. 
        In particular, we define~$b_{-B - 1} \doteq b_{n + B}$,~$b_{n + B + 1} = b_{-B}$, and for each~$t \notin \IntSet{-B}{n + B}$, we let~$\vec{x}[t] \doteq \vec{x}[ \PAREN{ \PAREN{t + B} \mod \paren{n + 2B + 1} } + \PAREN{-B} ]$.
    \end{lemma}

    Based on Lemma~\ref{lemma: bounds on inverting a single entry}, we have 
    
    \resizebox{\linewidth}{!}{
        \begin{minipage}{\linewidth}
        \begin{align*}
        \sum_{t \in \IntSet{-B}{n + B}} \card{\vec{x}[t]} 
            &\le \sum_{t \in \IntSet{-B}{n + B}} \frac{
                    e^{ -\eps } + e^{ \eps }
                }{
                    e^{ \eps } - e^{ -\eps }
                } \cdot \PAREN{ 
                    \begin{array}{c}
                        \ProbNorm \cdot \PAREN{
                            b_t + \frac{b_{t - 1} + b_{t + 1}}{e^{ -\eps } + e^{ \eps}}
                        }  + \\
                        \frac{e^{-\eps \cdot B }}{e^{ -\eps } + e^{ \eps } } \cdot \PAREN{ 
                            \card{ \vec{x} \bracket{t - B - 1} } + \card{ \vec{x} \bracket{t + B + 1} } + e^{ -\eps } \cdot \card{ \vec{x} \bracket{t - B} } + e^{ -\eps } \cdot \card{ \vec{x} \bracket{t + B} }
                        }
                    \end{array}
                } \\
            &= \sum_{t \in \IntSet{-B}{n + B}} \frac{
                    e^{ -\eps } + e^{ \eps }
                }{
                    e^{ \eps } - e^{ -\eps }
                } 
                \cdot \ProbNorm 
                \cdot b_t 
                \cdot \PAREN{ 1 + \frac{2}{e^{ -\eps } + e^{ \eps}} }
                +
                4 \cdot \sum_{t \in \IntSet{-B}{n + B}} \frac{
                    e^{ -\eps } + e^{ \eps }
                }{
                    e^{ \eps } - e^{ -\eps }
                } \cdot \frac{e^{-\eps \cdot B }}{e^{ -\eps } + e^{ \eps } } 
                \cdot \card{\vec{x}[t]} \\
            &= \sum_{t \in \IntSet{-B}{n + B}} \frac{
                    2 + e^{ -\eps } + e^{ \eps }
                }{
                    e^{ \eps } - e^{ -\eps }
                } 
                \cdot \ProbNorm 
                \cdot b_t 
                +
                \sum_{t \in \IntSet{-B}{n + B}} \frac{
                    4 \cdot e^{-\eps \cdot B }
                }{
                    e^{ \eps } - e^{ -\eps }
                } \cdot \card{\vec{x}[t]}. 
        \end{align*}
        \end{minipage}
    }
    This implies that 
    \begin{align*}
        \frac{
            e^{ \eps } - e^{ -\eps } - 4 \cdot e^{-\eps \cdot B }
        }{
            e^{ \eps } - e^{ -\eps }
        } \cdot \sum_{t \in \IntSet{-B}{n + B}} \card{\vec{x}[t]} 
        &\le \sum_{t \in \IntSet{-B}{n + B}} \frac{
                2 + e^{ -\eps } + e^{ \eps }
            }{
                e^{ \eps } - e^{ -\eps }
            } 
            \cdot \ProbNorm 
            \cdot b_t, \\
        \Longleftrightarrow
        \sum_{t \in \IntSet{-B}{n + B}} \card{\vec{x}[t]} 
        &\le \frac{
                2 + e^{ -\eps } + e^{ \eps }
            }{
                e^{ \eps } - e^{ -\eps } - 4 \cdot e^{-\eps \cdot B }
            } 
            \cdot \ProbNorm 
            \cdot \sum_{t \in \IntSet{-B}{n + B}} b_t, 
    \end{align*}
    
\end{proof}

\begin{proof}[Proof of Lemma~\ref{lemma: bounds on inverting a single entry}]
    Fix a~$t \in \IntSet{-B}{n + B}.$
    By assumption, we have 
    \begin{align*}
        \ProbNorm \cdot b_t 
        &\ge \ProbNorm \cdot \card{ \TransMatrix \bracket{ t, \cdot } \times \vec{x} } \\
        &= \card{
            \sum_{j \in [B]} e^{ -\eps \cdot j } \cdot \vec{x} \bracket{t - j} 
            + e^{ -\eps \cdot 0 } \cdot \vec{x} \bracket{ t } 
            + \sum_{j \in [B]} e^{ -\eps \cdot j } \cdot \vec{x} \bracket{t + j} 
        }.
    \end{align*}
    
    On the other hand, since~$\TransMatrix$ is a circulant matrix, 
    \begin{align*}
        \ProbNorm \cdot b_{t - 1} 
        &\ge  \ProbNorm \cdot \card{ \TransMatrix \bracket{ {t - 1}, \cdot } \times \vec{x} } \\
        &=\card{
            \sum_{j \in [B + 1]} e^{ -\eps \cdot (j - 1) } \cdot \vec{x} \bracket{t - j} 
            + e^{ -\eps } \cdot \vec{x} \bracket{t} 
            + \sum_{j \in [B - 1]} e^{ -\eps \cdot (j + 1) } \cdot \vec{x} \bracket{t  + j} 
        },   \\ 
        \ProbNorm \cdot b_{t + 1} 
        &\ge  \ProbNorm \cdot \card{ \TransMatrix \bracket{t + 1, \cdot } \times \vec{x} } \\
        &= \card{
            \sum_{j \in [B - 1]} e^{ -\eps \cdot (j + 1) } \cdot \vec{x} \bracket{t - j} 
            + e^{ -\eps } \cdot \vec{x} \bracket{t} 
            + \sum_{j \in [B + 1]} e^{ -\eps \cdot (j - 1) } \cdot \vec{x} \bracket{t  + j} 
        }.    
    \end{align*}
    
    Therefore 

    \vspace{-3mm}
    \resizebox{\linewidth}{!}{
      \begin{minipage}{1.12\linewidth}
        \begin{align*}
            \ProbNorm \cdot \PAREN{
                b_t + \frac{b_{t - 1} + b_{t + 1}}{e^{ -\eps } + e^{ \eps}}
            } 
            &\ge \ProbNorm \cdot \card{ 
                \TransMatrix \bracket{ t, \cdot } \times \vec{x}
                    - \frac{1}{e^{ -\eps } + e^{ \eps } } \cdot \PAREN{
                    \TransMatrix \bracket{ t - 1, \cdot } \times \vec{x} 
                    +
                    \TransMatrix \bracket{ t + 1, \cdot } \times \vec{x} 
                } 
            } \\
            &=  \card{ 
                \begin{array}{c}
                    - \frac{e^{-\eps \cdot B }}{e^{ -\eps } + e^{ \eps } }
                    \cdot \vec{x} \bracket{t - B - 1}
                    + \PAREN{1 - \frac{e^{ \eps }}{e^{ -\eps } + e^{ \eps } }} 
                    \cdot e^{-\eps \cdot B } \cdot  \vec{x} \bracket{t - B} 
                    + \PAREN{1 - \frac{2 e^{ -\eps }}{e^{ -\eps } + e^{ \eps }} } \cdot \vec{x} \bracket{ t } \\
                    + \PAREN{1 - \frac{e^{ \eps }}{e^{ -\eps } + e^{ \eps } }} 
                    \cdot e^{-\eps \cdot B } \cdot  \vec{x} \bracket{t + B}
                    - \frac{e^{-\eps \cdot B }}{e^{ -\eps } + e^{ \eps } }
                    \cdot \vec{x} \bracket{t + B + 1} 
                \end{array}
            } \\
            &\ge \frac{e^{ \eps } - e^{ -\eps }}{e^{ -\eps } + e^{ \eps }} \cdot \card{ \vec{x} \bracket{ t } } 
            -  \frac{e^{-\eps \cdot B }}{e^{ -\eps } + e^{ \eps } } \cdot \PAREN{ 
                \card{ \vec{x} \bracket{t - B - 1} } + \card{ \vec{x} \bracket{t + B + 1} } + e^{ -\eps } \cdot \card{ \vec{x} \bracket{t - B} } + e^{ -\eps } \cdot \card{ \vec{x} \bracket{t + B} }
            }.
        \end{align*}
      \end{minipage}
    }

    It concludes that 
    
    \vspace{-3mm}
    \resizebox{\linewidth}{!}{
      \begin{minipage}{1.15\linewidth}
        \begin{equation*}
            \card{ \vec{x} \bracket{ t } } 
            \le  
            \frac{
                e^{ -\eps } + e^{ \eps }
            }{
                e^{ \eps } - e^{ -\eps }
            } \cdot \PAREN{ 
                \ProbNorm \cdot \PAREN{
                    b_t + \frac{b_{t - 1} + b_{t + 1}}{e^{ -\eps } + e^{ \eps}}
                }  
                +
                 \frac{e^{-\eps \cdot B }}{e^{ -\eps } + e^{ \eps } } \cdot \PAREN{ 
                    \card{ \vec{x} \bracket{t - B - 1} } + \card{ \vec{x} \bracket{t + B + 1} } + e^{ -\eps } \cdot \card{ \vec{x} \bracket{t - B} } + e^{ -\eps } \cdot \card{ \vec{x} \bracket{t + B} }
                }
            }.
        \end{equation*}
      \end{minipage}
    }
\end{proof}

\section{Proof for Section~\ref{sec: lower bound}}
\label{appendix: sec: lower bound}

We show that, the lower bound in Theorem~\ref{theorem: accuracy lower bound of profile estimation} applies to our algorithms, 
if $\frac{1}{d} \, \big( \frac{e^\eps + 1}{e^\eps - 1} \big)^2 \in o \PAREN{ \frac{1}{\sqrt{d}} \cdot \frac{\eps}{c \cdot e^{c \cdot \eps}} }$, 
i.e., $\frac{1}{\sqrt{d}} \in o \PAREN{ \big( \frac{e^\eps - 1}{e^\eps + 1} \big)^2 \cdot \frac{\eps}{e^{c \cdot \eps} }}$.
We need the following lemma.

\begin{lemma}
    \label{lemma: sensitivity of the composition of discrete laplace and recovery}
    Let~$\cA$ be the discrete Laplace mechanism,~$C_A = \PAREN{\DiscreteLapNoise{e^{-\eps}}}^d$,~$F_\cA(\hist, C_A) = \hist + C_A$ and~$\cR$ be the profile reconstruction algorithm stated in Theorem~\ref{theorem: property of profile approximation algorithm}.
    Then
    $
        \Delta_{\cR \circ \cA} \in O \PAREN{
            \frac{1}{d} \cdot \norm{\TransMatrix^{-1}}_\infty
        }.
    $
\end{lemma}

Combined with Lemma~\ref{lemma: bounds on matrix norms} on the matrix norm, and the assumption of~$B$ in Theorem~\ref{theorem: property of profile approximation algorithm}, we have~$1 / d \, \| \TransMatrix^{-1} \|_\infty \in O \big( \frac{1}{d} \, \big( \frac{e^\eps + 1}{e^\eps - 1} \big)^2 \big)$, which proves our claim.

\begin{proof}[\bf Proof of Lemma~\ref{lemma: sensitivity of the composition of discrete laplace and recovery}]
    Let~$\hist, \hist' \in \IntSet{0}{n}^d$ be two neighboring histograms that differ in just one coordinate by at most~$1$, i.e.,~$\norm{\hist - \hist'}_1 \le 1$.
    For every possible realization~$c_A$ of~$C_A$, 
    the histograms~$\privHist = F_A \PAREN{ \hist, c_A }, \privHist' = F_A \PAREN{ \hist' , c_A }$ published by discrete Laplace mechanism also differ in just one coordinate by at most~$1$, i.e., 
    $$
        \norm{
            F_A \PAREN{ \hist, c_A } - F_A \PAREN{ \hist' , c_A }
        }_1
        = \norm{\hist - \hist'}_1
        \le 1. 
    $$
    Let~$\boundedProfileOfNoisyHist$ and~$\boundedProfileOfNoisyHist'$ be the profiles that corresponds to~$\privHist$ and~$\privHist'$, respectively. 
    Then~$\boundedProfileOfNoisyHist$ and~$\boundedProfileOfNoisyHist'$ differ in at most two coordinates, each by at most~$1 / d$.
    Applying~$\algoFstInv$ (Algorithm~\ref{algo: fast inversion}) to~$\boundedProfileOfNoisyHist$ and~$\boundedProfileOfNoisyHist'$ the amplify their distance by: 
    \begin{align*}
        \norm{
            \algoFstInv(\boundedProfileOfNoisyHist) - \algoFstInv(\boundedProfileOfNoisyHist')
        }_\infty
        &= \norm{ 
                \TransMatrix^{-1} \times \PAREN{
                    \boundedProfileOfNoisyHist - \boundedProfileOfNoisyHist'
                } -  \frac{
                    \inner{\vecTargetSupport}{
                        \TransMatrix^{-1} \times \PAREN{
                            \boundedProfileOfNoisyHist - \boundedProfileOfNoisyHist'
                        }
                    }
                }{ 
                    \inner{\vecTargetSupport}{\TransMatrix^{-1} \times \vec{a}}
                }  
                \cdot 
                \TransMatrix^{-1} \times \vec{a}
            }_\infty \\
        &\le \norm{\TransMatrix^{-1}}_\infty \cdot \PAREN{ 
                \norm{
                    \boundedProfileOfNoisyHist - \boundedProfileOfNoisyHist'
                }_\infty + \card{ 
                    \frac{
                        \inner{\vecTargetSupport}{
                            \TransMatrix^{-1} \times \PAREN{
                                \boundedProfileOfNoisyHist - \boundedProfileOfNoisyHist'
                            }
                        }
                    }{ 
                        \inner{\vecTargetSupport}{\TransMatrix^{-1} \times \vec{a}}
                    }
                } \cdot \norm{ \vec{a} }_\infty
            } \\
        &\le \norm{\TransMatrix^{-1}}_\infty \cdot \PAREN{ 
                \norm{
                    \boundedProfileOfNoisyHist - \boundedProfileOfNoisyHist'
                }_\infty +  \card{ 
                    \frac{
                        \inner{\vecTargetSupport}{
                            \TransMatrix^{-1} \times \PAREN{
                                \boundedProfileOfNoisyHist - \boundedProfileOfNoisyHist'
                            }
                        }
                    }{ 
                        \inner{\vecTargetSupport}{\TransMatrix^{-1} \times \vec{a}}
                    }
                }
            },
    \end{align*}
    where the last inequality follows since~$\norm{ \vec{a} }_\infty \le \norm{ \vec{a} }_p = 1$ for~$p = 1, 2, \infty$.
    We claim that 
    \begin{align*}
        \card{ 
            \frac{
                \inner{\vecTargetSupport}{
                    \TransMatrix^{-1} \times \PAREN{
                        \boundedProfileOfNoisyHist - \boundedProfileOfNoisyHist'
                    }
                }
            }{ 
                \inner{\vecTargetSupport}{\TransMatrix^{-1} \times \vec{a}}
            }
        } \le \frac{2}{d}.
    \end{align*}
    Since~$\boundedProfileOfNoisyHist - \boundedProfileOfNoisyHist'$ is a vector of at most two non-zero coordinates, each bounded by~$1 / d$, 
    \begin{align*}
        \card{
             \inner{\vecTargetSupport}{
                            \TransMatrix^{-1} \times \PAREN{
                                \boundedProfileOfNoisyHist - \boundedProfileOfNoisyHist'
                            }
                        }
        } 
        &= \frac{1}{d} \cdot \max_{t, t'} 
        \card{
            \inner{\vecTargetSupport}{                                  \TransMatrix^{-1}[\cdot, t]
            }
            +
            \inner{\vecTargetSupport}{                                  \TransMatrix^{-1}[\cdot, t']
            }
        } \\
        &\le \frac{2}{d} \cdot \max_{t} 
        \card{
            \inner{\vecTargetSupport}{                                  \TransMatrix^{-1}[\cdot, t]
            }
        }.
    \end{align*}
    On the other hand, since it holds for the standard basis vectors that~$\norm{\vec{e}_t}_p = 1$, 
    \begin{align*}
        \card{
              \inner{\vecTargetSupport}{\TransMatrix^{-1} \times \vec{a}}
        } 
        &\ge \max_{t} \card{
            \inner{
                \vecTargetSupport
            }{
                \TransMatrix^{-1} \times \vec{e}_t
            }
        } 
        = \max_{t} 
        \card{
            \inner{\vecTargetSupport}{\TransMatrix^{-1}[\cdot, t]}
        }.
    \end{align*}

    It remains to verify that~$\max_{t} 
        \card{
            \inner{\vecTargetSupport}{                                  \TransMatrix^{-1}[\cdot, t]
            }
        } \neq 0$.
    This is true, as 
    \begin{align*}
        \inner{\vecTargetSupport}{                                      \sum_{t} \TransMatrix^{-1}[\cdot, t]
        } 
        &= \inner{
                \vecTargetSupport
            }{
                \TransMatrix^{-1} \times \vec{\mathbbm{1}}
            } 
        = \inner{
                \vecTargetSupport
            }{
                \vec{\mathbbm{1}}
            } 
        = n + 1,
    \end{align*}
    where the equation~$\TransMatrix^{-1} \times \vec{\mathbbm{1}} = \vec{\mathbbm{1}}$ holds since~$\TransMatrix \times \vec{\mathbbm{1}} = \vec{\mathbbm{1}}$.

    Combined, we have
    \begin{align*}
        \norm{
            \algoFstInv(\boundedProfileOfNoisyHist) - \algoFstInv(\boundedProfileOfNoisyHist')
        }_\infty
        &\in O \PAREN{ 
            \frac{1}{d} \cdot \norm{\TransMatrix^{-1}}_\infty 
        }.
    \end{align*}
    To ease the notation, from now on, denote~$\estProfile = \algoFstInv(\boundedProfileOfNoisyHist)$ and~$\estProfile\,' \doteq \algoFstInv(\boundedProfileOfNoisyHist\,')$.
    We will finally show that, Applying~$\algoRounding$ (Algorithm~\ref{algo: rounding}) to~$\estProfile$ and $\estProfile\,'$ the amplify their $\ell_\infty$ distance only by some constant: 
    \begin{align*}
        \norm{
            \algoRounding \PAREN{ \estProfile }
            - 
            \algoRounding \PAREN{ \estProfile\,' }
        }_\infty
        &\in O \PAREN{ 
            \norm{
                \estProfile - \estProfile\,'
            }_\infty
        }.
    \end{align*}
    We investigate the steps of Algorithm~\ref{algo: rounding}.
    First, rounding the coordinates of~$\estProfile$ and~$\estProfile\,'$ to the interval~$[0, 1]$ only decreases their~$\ell_\infty$ distance.

    Next, the algorithm searches for some thresholds~$\tau$ and~$\tau'$, such that 
    \begin{align*}
        \sum_{t \in \IntSet{0}{n}} \PAREN{ \estProfile[t] - \min \set{\tau, \estProfile[t]} } &= 1, 
        \\
        \sum_{t \in \IntSet{0}{n}} \PAREN{ \estProfile\,'[t] - \min \set{\tau', \estProfile\,'[t]} } &= 1.
    \end{align*}
    We claim that~$\tau\,' \le \tau + \norm{
                \estProfile - \estProfile\,'
            }_\infty$,
    as 
    \begin{align*}
        \sum_{t \in \IntSet{0}{n}} \PAREN{ \estProfile\,'[t] - \min \set{\tau + \norm{
                \estProfile - \estProfile\,'
            }_\infty, \estProfile\,'[t]} } 
        \le 
        \sum_{t \in \IntSet{0}{n}} \PAREN{ \estProfile[t] - \min \set{\tau, \estProfile[t]} } 
        = 1.
    \end{align*}
    Via symmetry, we can prove that ~$\tau \le \tau' + \norm{
                \estProfile - \estProfile\,'
            }_\infty$.
    As Algorithm~\ref{algo: rounding} eventually update each~$\estProfile[t]$ and~$\estProfile'[t]$ by
    \begin{align*}
        \estProfile[t] \leftarrow \estProfile[t] - \min \set{\tau, \estProfile[t]},
        \quad 
        \estProfile\,'[t] \leftarrow \estProfile\,'[t] - \min \set{\tau', \estProfile\,'[t]},
    \end{align*}
    after the update, it still holds that~$\card{\estProfile[t] - \estProfile\,'[t]} \in O(\norm{
                \estProfile - \estProfile\,'
            }_\infty)$, which finishes the proof.
\end{proof}

\begin{proof}[\bf Proof of Theorem~\ref{theorem: accuracy lower bound of profile estimation}]
    We will prove the theorem by contraction. 
    Assume there exists $\eps$-DP and updatable algorithm~$\cA : \IntSet{0}{n}^d \rightarrow \cY$, profile estimation algorithm~$\cR : \cY \rightarrow [0, 1]^{n + 1}$, 
    such that, for every input~$\hist \in \IntSet{0}{n}^d$, s.t.
    \begin{equation*}
            \mathbb{E} [ \Vert \profile - \cR \paren{ \cA(\hist) } \Vert_\infty ] 
                \in o \bigparen{ {1} /  \paren{\sqrt{d} \cdot e^{c \cdot \eps}} }.
    \end{equation*}
    We will construct an~$\eps$-DP protocol~$P$ such that such that for every~$\vec{x}$ and~$\vec{y}$,
    Bob's output (denoted by $m_B$) satisfies 
    \begin{equation}
            \P{ \card{m_B - \inner{\vec{x}}{\vec{y}}} 
                \ge \err{} 
            } 
            \in o(1), 
    \end{equation}
    where $\err{} = \sqrt{d} / \paren{2 \cdot c \cdot e^{c \cdot \eps} }$,
    contradicting Fact~\ref{fact: two party accuracy lower bound}.
    
    The protocol is given by Algorithm~\ref{algo: two party protocol}. 
    Alice initializes the computation by sending~$\cA \paren{ \vec{x} + \vecAllOne}$ to Bob. 
    On receiving~$\cA \paren{ \vec{x} + \vecAllOne}$, Bob updates it with their own data~$\vec{y} + \vecAllOne$ by running algorithm~$\cU$, and the result is then passed to the profile estimation algorithm~$\cR$.
    Observe that by the property of updatable algorithm,~$\cU \paren{ \cA(\vec{x} + \vecAllOne), \vec{y} + \vecAllOne) } = \cA(\vec{x} + \vecAllOne + \vec{y} + \vecAllOne)$.
    Hence, $\cR \PAREN{ \cU \paren{ \cA(\vec{x} + \vecAllOne), \vec{y} + \vecAllOne) } }$ is essentially an estimate of the profile~$\profile$ of~$\hist \doteq \vec{x} + \vecAllOne + \vec{y} + \vecAllOne$.
    Finally, Bob publishes~$d \cdot \PAREN{ \estProfile[4] + \estProfile[0] - \estProfile[2] + 3 \cdot \Delta_{\cR \circ \cA} \cdot \LapNoise{1 / \eps} }$.
    where~$\LapNoise{1 / \eps} \in \R$ is a random variable following the Laplace distribution whose probability density function is given by~$
        p(z) = \frac{\eps}{2} \cdot \exp \PAREN{ - \eps \cdot \card{z} },
    $ 
    for all~$z \in \R.$
    
    It remains to analyze the privacy and utility guarantee of~$P$ respectively. 

    {\it Privacy Guarantee.}
    The view of Bob is~$\cA(\vec{x} + \vecAllOne)$, which is~$\eps$-DP by the assumptions of~$\cA$.
    The view of Alice is~$d \cdot \PAREN{ \estProfile[4] + \estProfile[0] - \estProfile[2] + 3 \cdot \Delta_{\cR \circ \cA} \cdot \LapNoise{1 / \eps} }$. 
    To prove this is~$\eps$-DP, it suffices to show that~$\estProfile[4] + \estProfile[0] - \estProfile[2]$ has sensitivity~$\Delta_{\cR \circ \cA}$. 
    
    Assume the randomness~$C_A$ of~$\cA$ has realization~$c_A$.
    The message~$\cA(\vec{x} + \vecAllOne)$ received by Bob therefore can be written as~$F_A(\vec{x} + \vecAllOne, c_A).$
    After performing the update~$\cU \paren{ \cA(\vec{x} + \vecAllOne), \vec{y} + \vecAllOne) } = \cA(\vec{x} + \vecAllOne + \vec{y} + \vecAllOne)$, the message becomes~$F_A(\vec{x} + \vecAllOne + \vec{y} + \vecAllOne, c_A)$. 
    Suppose that~$\vec{y}$ is replaced by a neighbor binary vector~$\vec{y}\,'$ such that~$\norm{\vec{y} - \vec{y}\,'}_1 \le 1$.
    Then~$\norm{\vec{x} + \vecAllOne + \vec{y} + \vecAllOne - (\vec{x} + \vecAllOne + \vec{y} + \vecAllOne\,')}_1 \le 1$, 
    and via the definition of~$\Delta_{\cR \circ \cA}$, we have ~$\norm{ \cR \PAREN{ F_A(\vec{x} + \vecAllOne + \vec{y} + \vecAllOne, c_A) } - \cR \PAREN{ F_A(\vec{x} + \vecAllOne + \vec{y}' + \vecAllOne, c_A) } }_\infty \le \Delta_{\cR \circ \cA}$.

    {\it Accuracy Guarantee.}
    Since~$\vec{x}$ and~$\vec{y}$ are binary vectors, $\inner{\vec{x}}{\vec{y}} = \card{ \set{ \ell \in [d] : \paren{ \vec{x} + \vec{y} }[\ell] = 2} } = d \cdot \profile[2]$.
    Therefore 
    \begin{align}
        &\card{
            d \cdot \PAREN{ \estProfile[4] + \estProfile[0] - \estProfile[2] + 3 \cdot   \Delta_{\cR \circ \cA} \cdot \LapNoise{1 / \eps} }
            - d \cdot \PAREN{ 
                \profile[2] + \profile[0] - \profile[2]
            }
        } \\
        &\le d \cdot \PAREN{ \card{
                \estProfile[4] - \profile[4]
            } + \card{
                \estProfile[2] - \profile[2]
            } + \card{
                \estProfile[0] - \profile[0]
            } 
            + 3 \cdot \Delta_{\cR \circ \cA} \cdot \card{ \LapNoise{1 / \eps} }
        }.
    \end{align}
    Since $\E{\card{ \LapNoise{1 / \eps} }} \in O( 1 / \eps )$, and by assumption, 
    $
        \E{ \norm{
                \estProfile - \profile
            }_\infty } 
        \in o \PAREN{ \frac{1}{\sqrt{d}} \cdot \frac{1}{e^{c \cdot \eps}} },
    $
    and
    $
        \Delta_{\cR \circ \cA} \in o \PAREN{ \frac{1}{\sqrt{d}} \cdot \frac{\eps}{e^{c \cdot \eps}} },
    $
    it concludes that 
    
    \vspace{-3mm}
    \resizebox{\linewidth}{!}{
      \begin{minipage}{1\linewidth}
        \begin{equation*}
            \E{ 
                \card{
                    d \cdot \PAREN{ \estProfile[4] + \estProfile[0] - \estProfile[2] + 3 \cdot   \Delta_{\cR \circ \cA} \cdot \LapNoise{1 / \eps} }
                    - d \cdot \PAREN{ 
                        \profile[2] + \profile[0] - \profile[2]
                    }
                }
            } 
            \in o \PAREN{ \frac{\sqrt{d}}{e^{c \cdot \eps}} }.
        \end{equation*}
      \end{minipage}
    }

    Combing $d \cdot \paren{\profile[4] + \profile[0] - \profile[2]} = \inner{\vec{x}}{\vec{y}}$ with Markov inequality, we see 
    \begin{align} 
        &\P{ 
            \card{
                d \cdot \PAREN{ \estProfile[4] + \estProfile[0] - \estProfile[2] + 3 \cdot   \Delta_{\cR \circ \cA} \cdot \LapNoise{1 / \eps} }
                - \inner{\vec{x}}{\vec{y}}
            }
            \ge \frac{ 2 \cdot c_1 \cdot \sqrt{d} }{ e^{c \cdot \eps} }
        } \\
        &\le { \E{ 
            \card{
                    d \cdot \PAREN{ \estProfile[4] + \estProfile[0] - \estProfile[2] + 3 \cdot   \Delta_{\cR \circ \cA} \cdot \LapNoise{1 / \eps} }
                    - \inner{\vec{x}}{\vec{y}}
                }
            } 
        } \Bigg/ {
            \PAREN{ \frac{ \sqrt{d} }{ 2 \cdot c \cdot e^{c \cdot \eps} } }
        } \\
        &\in o \PAREN{
            \frac{\sqrt{d}}{e^{c \cdot \eps}} 
            \middle/
            \PAREN{ \frac{ \sqrt{d} }{ 2 \cdot c \cdot e^{c \cdot \eps} } }
        } \\
        &= o(1).
    \end{align}
\end{proof}

\end{document}